\documentclass[twocolumn,12pt]{article}

\usepackage[american]{babel}
\usepackage{tikz}
\usepackage{natbib}
\newcommand{\ignore}[1]{}

\ignore{\title{\bf Database Querying under Missing Values  Governed by Missingness Mechanisms}

\author{{\bf Leopoldo Bertossi}^1, {\bf Farouk Toumani}^2 and {\bf Maxime Buron}^2\\
    ^1Carleton University,
    Ottawa, Canada \& IMFD, Chile.\\
    ^2LIMOS, CNRS, UCA, France.}
}

\usepackage{amsmath}
\usepackage{amssymb}
\usepackage{graphicx}
\usepackage{color}
\usepackage{hhline}
\usepackage{tabu}
\usepackage{multirow}
\usepackage{caption}

\usepackage{tabu}

\usepackage{hyperref}

\newcommand{\boxtheorem}{\hfill $\blacksquare$\\}
\newcommand{\nit}[1]{{\it #1}}
\DeclareMathOperator*{\argmax}{\mathsf{argmax}}
\DeclareMathOperator*{\argmin}{\mathsf{argmin}}

\newcounter{theorem-counter}
\newcounter{corollary-counter}
\newcounter{lemma-counter}
\newcounter{definition-counter}
\newcounter{example-counter}
\newcounter{proposition-counter}
\newcounter{remark-counter}
\newcounter{problem-counter}
\newcounter{claim-counter}

\newenvironment{theorem}%
{\vskip \abovedisplayskip \refstepcounter{theorem-counter}%
\noindent {\bf Theorem \arabic{theorem-counter}.}}%

\newenvironment{corollary}%
{\vskip \abovedisplayskip \refstepcounter{corollary-counter}%
\noindent {\bf Corollary \arabic{corollary-counter}.}}%

\newenvironment{lemma}%
{\vskip \abovedisplayskip \refstepcounter{lemma-counter}%
\noindent {\bf Lemma \arabic{lemma-counter}.}}%

 \newenvironment{definition}%
 {\vskip \abovedisplayskip \refstepcounter{definition-counter}%
 \noindent {\bf Definition \arabic{definition-counter}.}}%

\newenvironment{example}%
{\vskip \abovedisplayskip \refstepcounter{example-counter}%
\noindent {\bf Example \arabic{example-counter}.}}%

{\vskip \abovedisplayskip \refstepcounter{exampleAp-counter}%
\noindent {\bf Example \ref{app:examples}.\arabic{exampleAp-counter}.}}%

\newenvironment{proposition}%
{\vskip \abovedisplayskip \refstepcounter{proposition-counter}%
\noindent {\bf Proposition \arabic{proposition-counter}.}}%

\newenvironment{remark}%
{\vskip \abovedisplayskip \refstepcounter{remark-counter}%
\noindent {\bf Remark \arabic{remark-counter}.}}%

{\vskip \abovedisplayskip \refstepcounter{problem-counter}%
\noindent {\bf Problem \arabic{problem-counter}.}}%

{\vskip \abovedisplayskip \refstepcounter{claim-counter}%
\noindent {\bf Claim \arabic{claim-counter}.}}%

\newcounter{propositionA-counter}
\newcounter{lemmaA-counter}
\newcounter{exampleAp-counter}

{\vskip \abovedisplayskip \refstepcounter{lemmaA-counter}%
\noindent {\bf Lemma A.\arabic{lemmaA-counter}.}}%

{\vskip \abovedisplayskip \refstepcounter{propositionA-counter}%
\noindent {\bf Proposition A.\arabic{propositionA-counter}.}}%

\newcommand{\mc}[1]{\mathcal{ #1}}

\newcommand{\mbb}[1]{\mathbb{ #1}}
\newcommand{\mbf}[1]{\mathbf{ #1}}

\newcommand{\C}[1]{\mathcal{C}}
\newcommand{\T}[1]{\mathcal{T}}

\newcommand{\na}{{\bf \sf na}}

\newcommand{\Nn}{\mbox{\scriptsize \sf NULL}}

\newcommand{\red}[1]{\textcolor{red}{#1}}

\newcommand{\blue}[1]{\textcolor{blue}{#1}}

\newcommand\independent{\protect\mathpalette{\protect\independenT}{\perp}}
\def\independenT#1#2{\mathrel{\rlap{$#1#2$}\mkern2mu{#1#2}}}

\newcolumntype{M}[1]{D{.}{.}{1.#1}}

\title{{\bf \Huge Database Querying under Missing Values  Governed by Missingness Mechanisms}

\vspace{2mm}
{\bf Leopoldo Bertossi}$^1$, {\bf Farouk Toumani}$^2$ and {\bf Maxime Buron}$^2$

    $^1$ Carleton University,
    Ottawa, Canada \& IMFD, Chile.\\
    $^2$ LIMOS, CNRS, UCA, France.}

    \date{}

\begin{document}
\maketitle
\pagestyle{plain}
\thispagestyle{empty}

\begin{abstract}
We address the problems of giving a semantics
to- and doing query answering (QA) on a relational
database (RDB) that has missing values (MVs). The
causes for the latter are governed by a Missingness
Mechanism that is modelled as a Bayesian Network, which represents a Missingness Graph (MG) and
involves the DB attributes. Our approach considerable
departs from the treatment of RDBs with NULL
(values). The MG together with the
observed DB allow to build a block-independent
probabilistic DB, on which basis we propose two
QA techniques that jointly capture probabilistic uncertainty
and statistical plausibility of the implicit
imputation of MVs. We obtain complexity
results that characterize the computational feasibility
of those approaches.
\end{abstract}

\vspace{-2mm}\section{Introduction}\label{sec:intro}

\vspace{-2mm}
It is common to find missing values (MVs) in a  relational database (DB) $D^\star$, that is, values for some attributes that are not reported. We will  use \na, for ``not available", to indicate that the true value of a variable in that place is  absent. No other semantics is assigned to  \na.
\ We consider the {\em observed}  DB $D^\star$ as  obtained from an independent  sample from an outside reality.\footnote{This kind of observed DBs, including independence and MVs,  are common with census-related data \cite{gelman}.} We assume that  $D^\star$  is an incomplete representation of the ``true", possibly unknown and {\em underlying} DB $D$ that represents the external reality. All of $D$'s attributes have non-\na \ values that we may not see.

\vspace{-2mm}
\begin{example}\label{ex:first}  Relation $R^\star$  in Table \ref{table:first} belongs to  an observed DB $D^\star$. We will use global tuple identifiers (tids): $\tau_1$, etc. Some (true) values are missing, as indicated with \na. Following notation and conventions in \citep{mohan21}, \ attribute \  $B^o$ \ is   fully observable, i.e. it never shows  \na, whereas an attributes  $A^\star, C^\star$ may have MVs.  $R^\star$ has the {\em observed schema} $R^\star(A^\star,B^o,C^\star)$.

\vspace{-5mm}
\begin{center}{\scriptsize
\captionof{table}{Relation with Missing Values.}\label{table:first}
$\begin{tabu}{c|c|c|c|}\hline
R^\star & A^\star & B^o & C^\star \\ \hline
\tau_1 & a_1 & 0& c_1\\
\tau_2 & a_2& 1& \na\\
\tau_3 &\na & 0&c_3\\
\tau_4 & a_4&0&c_4\\
\tau_5 & \na&1&c_5\\
\tau_6 &\na&1& \na\\ \hhline{~---}
\end{tabu}$
}
\end{center}

\vspace{-3mm}

\noindent         The underlying  DB $D$ has the {\em underlying schema} $R(A^m,B^o,C^m)$, where the superscript ``$m$''  indicates that the attribute may have MVs. We do not observe $A^m$ directly, but via $A^\star$. When the latter does not show an \na, it shows the real value of $A^m$.
\boxtheorem
 \end{example}

\vspace{-5mm}The main question that underlies our work is the following: \ {\em How can we query the ``underlying" DB $D$ in a meaningful way ``through" the observed DB $D^\star$?} \  This quest looks hopeless unless we have some additional information. Actually, we will assume we know about the {\em causes for the occurrence of MVs} by means of a Bayesian Network (BN) \cite{pearl} that stochastically models, as a {\em Missingness Graph} (MG), the  {\em Missingness Mechanisms} (MMs) at play \cite{rubinBook}. Those MMs describe under what stochastic conditions MVs may appear.

\vspace{-5mm}
\begin{center}
    \includegraphics[width=2.7cm]{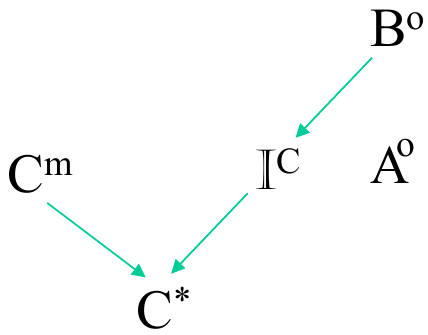}
\vspace{-2mm}\captionof{figure}{Missingness Graph.}\label{fig:mg3Intro}
\end{center}

\vspace{-6mm}\begin{example}\label{ex:second} (Ex. \ref{ex:first} cont.) \  For simplicity, assume now that attribute $A$ is fully observed and denoted with $A^o$. The MG in Figure \ref{fig:mg3Intro} shows $A^o$ as
independent from the other attributes. \ $\mbb{I}^C$ is an indicator variable for $C^m$ that
takes  the value $1$ when the value for $C^m$ is missing, and $0$,
otherwise. The MG tells us that whether $C^\star$ takes the value \na \ or not, depends on variable $C^m$ and $B^o$, for the latter, via $\mbb{I}^C$. \boxtheorem
\end{example}

In general, we start with an observed DB, $D^\star$, possibly with MVs, and a MG, $\mc{M}$, that models the occurrence of MVs. Observed values are assumed to be certain, but MVs are   uncertain. We also assume $D^\star$ is compatible with the joint distribution induced by the MG.\footnote{Learning a MG and checking an instance against a MG are orthogonal problems to ours, and matter of ongoing separate research.}

We will extract \na-free and stochastically reliable information  from $D^\star$ by using the
combination of $D^\star$ and $\mc{M}$ to create  (virtually or physically) a {\em Block-Independent Probabilistic DB} (BID) \cite{suciu}.
This BID, in its turn, gives rise to  a collection $\mc{W}$ of {\em possible worlds}, each of them an \na-free DB instance for the underlying schema, with an associated global probability.   $\mc{W}$ represents the uncertainty due to the occurrence of MVs. Each possible world in it becomes a possible materialization -without MVs- of the partially observed underlying DB $D$. \ignore{Due to the implicit and stochastic imputation of MVs, possible worlds may have duplicate tuples, which becomes relevant in our setting.}

Our approach can be seen as a generalization of traditional {\em imputation techniques} \cite{gelman}, where instead of a single ``clean" instance, a class of probabilistically-weighted clean instances is considered. This makes for a more principled and uncertainty-aware data semantics and query answering (QA). Our way of dealing with MVs can be understood as a form of collective imputation  by means of  probability distributions over possible values; distributions that are related to each other through an underlying joint distribution determined by a MG. Those ``cell-level" distributions can be seen as footprints of that joint distribution.

Query answering (QA) can be done on the BID, which has a clear semantics \cite{suciu}. Without neglecting this direct approach,   we propose an alternative semantics that builds on the BID, but takes into account the {\em statistical compliance} of possible worlds with the distribution induced by the MG.   We start by collecting in different {\em classes} those possible worlds that are essentially the same, the  {\em matching worlds}. QA on matching worlds returns the same answer.  After that, we introduce a {\em measure of compliance} for classes of matching worlds. We  identify those classes that are maximally compliant, and do QA on top of them. \  More specifically, in this work we make the following contributions:

1. We define the semantics of a DB with  MVs in relation to a MG representing the missingness mechanisms at play. \ This is done through  a BID associated to the observed DB and the MG, giving rise to classes of matching possible worlds. Classes are used to define QA.

    2. We propose two data and QA semantics:  The \emph{Most-Probable Classes} (MPC), and the \emph{Most-Compliant Classes} (MCC). The latter has a statistical component in that compliance is measured as the distance between a world's empirical distribution and that induced by the MG.

    3. We investigate the computational complexity of the  QA-semantics. For a broad and common family of notions of compliance, one can compute a most-compliant class in polynomial time in the size of the observed DB, on which QA can be done. Furthermore, although there may be exponentially many classes,  enumerating all the most-compliant classes and query answers on them  can be done with polynomial-time delay.  We obtain several hardness results for the MCC-semantics. Still, its has better computational properties than the MPC-semantics.

This paper is structured as follows. Section~\ref{sec:prel} provides background. Section~\ref{sec:mgs} introduces MMs and MGs. Section \ref{sec:genBIDs} introduces the BIDs associated to an observed DB. Section~\ref{sec:queries} introduces classes of matching worlds and class-based QA. Section \ref{sec:compliance-pw} introduces the notion of compliance.  In Section \ref{sec:complexity}, we investigate algorithmic and complexity aspects of class computation and QA. \ignore{Section \ref{sec:balanced} investigates algorithmic aspects of a particular case of generated BIDs.} Section \ref{sec:relWork} discusses related work; and    Section~\ref{sec:conclusion}  summarizes our contributions and outlines directions of future work. The Appendix  provides additional material and proofs of results.

\vspace{-3mm}
\section{Background and\\ Preliminaries}
\label{sec:prel}

\vspace{-2mm}
\paragraph{Relational Databases.}
A relational schema, usually denoted with $S$,  is a finite set of logical predicates, usually denoted with $R$,  with fixed arities. Variables, a.k.a. attributes or features, are associated to the predicate positions. A relational DB, usually denoted with $D$,  is a collection of relations which are finite extensions for the predicates. Their elements are called tuples; and their values come from fixed and finite domains. $\nit{dom}(X)$ denotes the domain of an attribute $X$. For several attributes, $\bar{X}$, $\nit{dom}(\bar{X})$ denotes the cartesian product of the individual domains.  We  use $\nit{dom}$ for the union of the domains. \ The string \na, indicating a missing value, \ does not belong to any of these domains. $\nit{dom}^{\!\star}$ denotes $\nit{dom} \cup \{\na\}$.
 \ Unless otherwise stated, DB relations may have duplicates, i.e. repeated tuples, which we tell apart by means of global {\em tuple identifiers} (tids) that appear in a first attribute of tuples, acting as a surrogate key, as shown  in Example \ref{ex:first}. We denote tids with $\tau, \tau_1, \ldots$. We will frequently denote and refer to tuples by their tids. \ignore{Beware: With $\tau\!\!\!\downarrow$ we denote the projection of the tuple identified by $\tau$ on the non-tid attributes. In Example \ref{ex:first}, $\tau_2\!\!\downarrow = \langle a_2, 1, \na\rangle$.}

 We concentrate on  {\em conjunctive queries} (CQs), i.e. ``project-join-select" queries,  \ignore{Beware: existential quantifications of conjunctions of $L(S)$-atoms, possibly with built-ins.  Boolean CQs (BCQs) have all variables  quantified.} and aggregations over CQs. \ignore{Beware:  scalar or with group-by, with aggregation functions $\nit{sum}, \nit{avg}, \nit{count}, \nit{count*}, \nit{max}, \nit{min}$.} Queries may have constants from $\nit{dom}$, and then, different from \na.
The set of answers to a query $\mc{Q}$ from a DB $D$ is denoted with $\mc{Q}[D]$.
Boolean queries (BQ) have $0$ or $1$ answers.

\vspace{-2mm}
\paragraph{Probabilistic Databases.}
  Given a relational DB, $D$, a PDB, $D^p$ associated to $D$ is a {\em collection $\mc{W}$ of possible worlds} that are the subinstances $W$ of  $D$. Each $W \in \mc{W}$ has a  probability $p(W)$, such that  $\sum_{W\in\mc{W}}p(W) = 1$ \cite{suciu}. $D^p$  becomes a discrete {\em probability space} $\langle \mc{W}, p\rangle$, with $p$ probability distribution on $\mc{W}$. \ The probability of tuple $\tau$ being in $D$ (as seen through $D^p$) is  \
$P(\tau) \ := \sum_{\tau \in W \in \mc{W}} p(W)$.
\ A numeric query $\mc{Q}$ on $D$ becomes a random variable on  $\mc{W}$; with Bernoulli distribution if it is a BQ,  in which case, the probability of $\mc{Q}$ (being true) is \
$P(\mc{Q})  := P(\mc{Q} = 1) :=
\sum_{W \in \mc{W}: \ W \models \mc{Q}} p(W)$. \ignore{
If $\mc{Q}(\bar{x})$ is an open query, and $\bar{a}$, a sequence of constants, \ $P(\bar{a}) := P(\mc{Q}[\bar{a}])$.
\ Under this semantics,  each answer comes with a probability (of being an answer).}

\vspace{-0.5cm}
\begin{center}
\captionof{table}{A BID.}\label{tab:pdbs}
{\scriptsize
$\begin{tabu}{c|c|c|c||c|}\hline
R & A & B & C& P \\ \hline
\tau^1_1& a_1 & b_1& c_1&p_1^1\\
\tau^2_1 & a_2& b_2& c_2&p_1^2\\
\hline
\tau^1_2 & a_3& b_3& c_3&p_2^1\\
\tau^2_2& a_4& b_4& c_4&p_2^2\\
\tau^3_2& a_5& b_5& c_5&p^3_2\\
\hline
\end{tabu}$}
\end{center}

\vspace{-2mm}In our work we are interested in {\em block-independent} PDBs (BIDs) \cite{suciu}.  Tuples have probabilities, and tuples of a same relation are separated in mutually independent blocks of mutually exclusive tuples. Tuples in a block are stochastically independent from tuples in other blocks. Table \ref{tab:pdbs} shows a BID $R$ with two blocks.

 In a possible world $W \in \mc{W}$, the corresponding  relation $R_W$ is built from relation $R$,
 by choosing from each block $\mc{B}_j$ exactly one tuple $\tau_j^i$ . The probability associated to $R_W$ is defined by: \ $p(R_W) := \Pi p_j^i$. Any other kind of subrelation has probability $0$. The probability of a world $W$ is the product of the probabilities of its relations. A  \emph{most probable database} (MPDB) is a possible world with the highest probability. Within a block, the sum, $p$, of the tuples' probabilities is not greater than $1$. In case none of the tuples is chosen from the block, the block contributes with the factor $(1-p)$  to the relation's probability. A {\em tuple-independent PDB} (TID) is a particular kind of BID, where each block has a single tuple with a probability that can be less than $1$.

\vspace{-2mm}
\section{Missingness and\\ Observed DBs}\label{sec:mgs}

 \vspace{-2mm}A Missingness Mechanism (MM)
 specifies how the occurrence of MVs in a variable depends on the values that other (or the same) variables take.  \
There are some MMs that commonly  found in practice \cite{rubinBook}. Following  \cite{mohan21}, MMs are represented by {\em Missingness Graphs} (MGs), i.e. Bayesian Networks whose  directed edges capture causal or stochastic-dependence directions. Accordingly, MGs become directed acyclic graphs subject to the usual {\em Markov Condition}: Given its parents, a variable is independent from its non-descendant variables \cite{darwicheCommACM,pearl}.

\ignore{Beware:
\vspace{1mm}
\noindent {\bf (A)} \ The MVs for attribute $A^\star$  appear when, while we go down the column, we throw a die. If the subindex, $n$, of $a_n$ is equal to the number on the die (modulo 6), we record \na.  This is {\em Missingness Completely at Random} (MCAR) for $A^\star$.

\vspace{1mm}
\noindent {\bf (B)} \  The occurrence of MVs for $C^\star$ depends on the values that the fully observed variable $B^o$ takes. For example, we throw a die, but only when $B^o$ takes the value $1$. If the die shows an odd number, we record \na \ for $C^\star$. This is {\em Missingness at Random} (MAR) for $C^\star$.

\vspace{1mm}
\noindent {\bf (C)} \ In addition, and as opposed to the other two, we may also have {\em Missingness Not Completely at Random} (MNCAR). Under this mechanism, very complex  cases of dependencies may occur that determine the occurrence of MVs. A relevant particular case is that of {\em self-censorship}, where a variable determines its own MVs. For example,  a salary (randomly) shows as a MV  when the salary is high.

\subsection{Representing MMs and Data with MVs}
}

\vspace{-3mm}
\begin{example}  \label{ex:mgs} \ignore{(Ex. \ref{ex:second} cont.)
As in \cite{mohan21}, we introduce an {\em missingness indicator variable} $\mathbb{I}^C$ that takes the value $1$ exactly when $C$ is not fully observed, and $0$, otherwise. We denote with $C^m$ this true, partially observed variable. \  $C^\star$ is its observed version, i.e. that in Example \ref{ex:first}, including the \na's. All we see about $C^m$ is what  $C^\star$ shows: When $C^\star$ shows a non-\na \ value, it is the true value of $C^m$. When $C^\star$ shows \na, we do not know the true value of $C^m$.
\  The fully observed variable $B$ is usually denoted with $B^o$.}  The MG in Figure \ref{fig:mg}(a) corresponds to the MCAR case (missing completely at random) in that variable $\mathbb{I}^C$ does not depend on any other variable, and has absolute probabilities of taking values $0$ or $1$.

\vspace{-5.5mm}
 \begin{center}   \includegraphics[width=6.5cm]{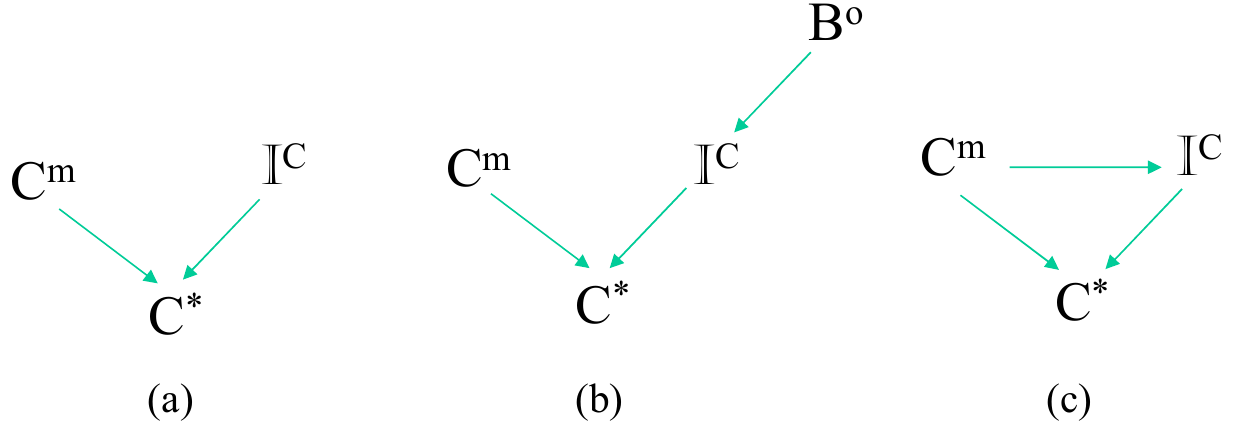}
\end{center}
\vspace{-6mm}\captionof{figure}{Three Missingness Graphs.}\label{fig:mg}

\vspace{2mm}
The MG in Figure \ref{fig:mg}(b) shows a case of MAR (missing at random):   $C^\star$ stochastically depends on $C^m$ and $\mathbb{I}^C$; while the latter depends on $B^o$.
  As expected, $C^\star$ depends directly on its real, underlying version, $C^m$, and its own  indicator variable. \ignore{We can even add a  {\em structural equation} at the sink node: \ $C^\star = \na \mbox{ when } \mathbb{I}^C=1$, and $C^\star = C^m, \mbox{ otherwise}$.} \ The BN also has the distributions: \ $P(C^m)$, $P(B^o)$, $P(\mathbb{I}^C|B^o)$, $P(C^\star|C^m,\mathbb{I}^C)$. Furthermore, due to the  Markov condition, absolute and conditional independences hold:
$B^o \independent C^m, \ \ \mathbb{I}^C \independent C^m, \ \ (C^\star \independent B^o)|\mathbb{I}^C$.

Finally, the MG in Figure \ref{fig:mg}(c) shows a case of MNCAR (missing not completely at random); actually of {\em self-censorship}:  $\mathbb{I}^C$ depends on $C^m$. For example, $C^m$ could be $\nit{Salary}^m$, that -with hight probability- is not reported when it is very high. When $\mbb{I}^\nit{Salary}$ takes value $1$, $\nit{Salary}^\star$ takes value $\na$. Otherwise, $\nit{Salary}^\star$ takes the real value of $\nit{Salary}^m$.  The BN has  absolute and conditional distributions: $P(C^m), P(\mathbb{I}|C^m), P(C^\star|C^m,\mathbb{I}^C)$.
\boxtheorem \end{example}

\vspace{-3mm}In a MG, each variable $X^\star$ is a sink, and has exactly two parents:  $X^m$ and $\mathbb{I}^X$.  $\mathbb{I}^X$ might depend on variables of the forms $Y^o$ and $\mathbb{I}^Z$, for other variables of the form  $Z^m$.  Variables of the form $B^o$ are sources.
\ In a MG $\mc{M}$ the  probability distributions at its nodes determine a {\em joint distribution} $P^\mc{M}$ over all variables. \ignore{For example, the MG in Figure \ref{fig:mg}(b) has the factorized joint distribution
$P^\mc{M}(B^o,C^\star,\mathbb{I}^C,C^m) = P(C^\star|C^m,\mbb{I}^C) \times P(\mbb{I}^C|B^o) \times P(C^m) \times P(B^o)$.}   In particular, a tuple $\tau$ for the underlying schema $S$ has a marginal probability $P^\mc{M}(\tau)$.
 Furthermore,  MGs satisfy the following conditions, where,  $c, c^\prime \in \nit{dom}(X^m)$ are arbitrary, but different:\ignore{\footnote{In first equations, we could use instead a value close to $1$ or $0$. Alternatively, we could use causal networks \cite{mohan21}, with deterministic {\em structural equations} at source nodes, e.g. ``$X^\star = \na \ \mbox{ {\em if} } \ \mbb{I}^X=1, \mbox{ {\em and} } \ X^\star = X^m \ \mbox{ {\em  if} } \ \mbb{I}^X=0$".}}\vspace{-6mm}

{\footnotesize \begin{eqnarray} P^\mc{M}(X^\star=\na~|~\mbb{I}^X=1,X^m=c, \ldots) &=& 1, \label{eq:ass1} \\  P^\mc{M}(X^\star=\na,\ldots~|~\mbb{I}^X=0,X^m=c,\ldots) &=& 0, \label{eq:ass2} \\
P^\mc{M}(X^\star=c^\prime,\ldots~|~\mbb{I}^X=0,X^m=c,\ldots) &=& 0,  \label{eq:ass2prime} \\
 P^\mc{M}(X^\star=c~|~\mbb{I}^X=0,X^m=c, \ldots) &=& 1, \label{eq:ass3}
 \end{eqnarray}}

\vspace{-9mm}{\footnotesize \begin{eqnarray}
 P^\mc{M}(X^\star=\na, \ldots\hspace{-3mm}&|&\hspace{-3mm}\mbb{I}^X=1,X^m=c,\ldots) = \nonumber\\ &&P^\mc{M}(\ldots~|~\mbb{I}^X=1,X^m=c, \ldots), \label{eq:ass4}\\
 P^\mc{M}(X^\star=c, \ldots&\hspace{-3mm}|&\hspace{-3mm}\mbb{I}^X=0, X^m=c, \ldots)  = \nonumber\\
 &&P^\mc{M}(\ldots~|~\mbb{I}^X=0,X^m=c, \ldots). \label{eq:ass5}
\end{eqnarray}}

\vspace{-6mm}
\begin{example} \label{ex:new} \ (Ex. \ref{ex:second} cont.) \ Consider the observed DB $R^\star$ in Table \ref{tab:dbex0}(a) for the observed schema $S^\star$.  MVs in it are governed by the MG $\mc{M}$ in Example \ref{ex:second}.
On the basis of the local distributions $P(A^o),$ $ P(C^\star|C^m,\mbb{I}^C), P(\mbb{I}|B^o), P(B^o)$, and  the Markov condition, $\mc{M}$ induces the factorized joint distribution:

\vspace{-6mm}
{\scriptsize \begin{eqnarray}
P^\mc{M}(A^o,B^o,C^\star,\hspace{-3mm}&&\hspace{-3mm}\mbb{I}^C,C^m)=P^\mc{M}(C^\star~|~C^m,\mbb{I}^C) \times P^\mc{M}(I^C~|~B^o)  \nonumber \\ &&\times P^\mc{M}(C^m)    \times \ P^\mc{M}(B^o) \times P^\mc{M}(A^o). \label{eq:joint}
\end{eqnarray}}

\vspace{-4mm}
\captionof{table}{(a) Observed  \ (b) Underlying  (c) Expanded DBs.}\label{tab:dbex0} \vspace{-2mm}
\begin{center}
 {\tiny
 $\begin{tabu}{c|c|c|c|}\hline
 R^\star & A^o & B^o & C^\star\\ \hline
\tau_1& a_1 & 0& c_1\\
\tau_2& a_2& 1& \na\\
\tau_3&a_3 & 0&c_3\\
\tau_4& a_4&0&c_4\\
\tau_5& a_5&1&c_5\\
\tau_6&a_6&1& \na\\ \hhline{~---}
\end{tabu}$~~~~~~~~~
$\begin{tabu}{c|c|c|c|}\hline
 W & A^o & B^o & C^m\\ \hline
 \tau_1& a_1 & 0& c_1\\
 \tau_2& a_2& 1& \underline{c_5}\\
\tau_3&a_3 & 0&c_3\\
\tau_4& a_4&0&c_4\\
\tau_5& a_5&1&c_5\\
\tau_6&a_6&1& \underline{c_2}\\ \hhline{~---}
 \end{tabu}$

 \vspace{1mm}
$\begin{tabu}{c|c|c|c|c|c|c|c|}\hline
R^\nit{ex} & A^o  & B^o & C^\star & \mathbb{I}^C&C^m\\ \hline
\tau_1& a_1 &  0& c_1 & 0&c_1\\
\tau_2& a_2&  1& \na & 1&c_5\\
\tau_3&a_3  &  0&c_3&0&c_3\\
\tau_4& a_4 & 0&c_4 &0&c_4\\
\tau_5& a_5 & 1&c_5 & 0&c_5\\
\tau_6&a_6& 1& \na &1&c_2\\ \hhline{~------}
\end{tabu}$}
 \end{center}

\vspace{-1mm}Table \ref{tab:dbex0}(b) shows $W$ as a possible \na-free underlying instance for the underlying schema $S$, where the true values $c_5$ and $c_2$ (underlined) for $C^m$ where not all observed, giving rise to $R^\star$.   The ``expanded" relation $R^\nit{ex}$ in Table \ref{tab:dbex0}(c), for an {\em expanded schema} $S^\nit{ex}$, shows a possible sample from the outside reality represented by $\mc{M}$, with values for all variables  in it. $R^\star$ and $R$ can be seen as {\em footprints} of $R^\nit{ex}$.     A ``user" has access only to the observed instance $R^\star$, but queries will be posed in the language  of the underlying schema $S$ (see Section \ref{sec:queries}).
 \boxtheorem
\end{example}

 \vspace{-4mm}{\em We will assume that in an observed instance $D^\star$ all the tuples are stochastically independent from each other, as obtained from an independent sample of the real world}.
In particular, $D^\star$ may have duplicate tuples. Furthermore,  we will assume that an observed DB $D^\star$ is {\em compliant} with the given MG $\mc{M}$. In particular,   $D^\star$ satisfies both the (in)dependencies represented by the MG and its induced distributions.

A possible underlying DB, such as $W$ in the previous example,  may also introduce  duplicates, when \na's are replaced by domain values. \ignore{For now, we make no assumption about the compliance of such an instance $W$ with $\mc{M}$. We retake this issue in Section \ref{sec:compliance-pw}.}

\vspace{-2mm}
\section{BIDs for Missing\\ Values}\label{sec:genBIDs}

\vspace{-2mm}
 We can see observed DB $D^\star$, together with a MG $\mc{M}$, as a single representation for a set, $\mc{W}(D^\star)$, of {\em possible worlds}, namely, of {\footnotesize \na}-free databases, $W$, for the underlying schema $\mc{S}$. For this, we start by replacing values from $\nit{dom}$ for each occurrence of an   \na \ in $D^\star$. Then,  each tuple with MVs in $D^\star$ gives rise to a  ``block" of  ``potentially observed" tuples, without \na, each of them with a probability. We obtain a BID, $D^p(D^\star,\mc{M})$. The following example shows this; in particular, how we assign probabilities to the tuples in a block. The general construction can be found in Definition \ref{def:blocks} in the Appendix.

\vspace{-3mm}
\begin{example}  (Ex. \ref{ex:new} cont.) \label{ex:new2}  Assume $\nit{dom}(C)= \{c_1, \ldots, c_5\}$, Table \ref{tab:dbex0}(a) gives rise to the relation in Table \ref{tab:blocks}(a).  The first tuple, fully observed, gives rise to a {\em certain} block with one tuple,  with probability $1$.

\captionof{table}{(a) Blocks of tuples,$\qquad$ (b) A BID.}\label{tab:blocks}
 \vspace{-2mm}\begin{center}
 {\tiny
$\begin{tabu}{rr|c|c|c|}\hhline{~----}
&R & A^o & B^o & C^m\\ \hhline{~----}
&\tau_1& a_1 & 0& c_1\\
\hhline{-----}
&\tau_2^1& a_2& 1& c_1\\
&\tau_2^2& a_2& 1& c_2\\
\mc{B}(\tau_2)\hspace*{-2mm}&\tau_2^3& a_2& 1& c_3\\
&\tau_2^4& a_2& 1& c_4\\
&\tau_2^5& a_2& 1& c_5\\
\hhline{-----}
&\cdots& \cdots&\cdots&\cdots
\\ \hhline{~~---}
\end{tabu}~~\begin{tabu}{c|c|c|c||c|}\hline
R^p & A^o & B^o & C^m& p^\nit{BID}\\ \hline
\tau_1& a_1 & 0& c_1&1\\
\hline
\tau_2^1& a_2& 1& c_1&p_2^1\\
\tau_2^2& a_2& 1& c_2&p_2^2\\
\tau_2^3& a_2& 1& c_3& p_2^3\\
\tau_2^4& a_2& 1& c_4&p_2^4\\
\tau_2^5& a_2& 1& c_5&p_2^5\\
\hline
\cdots& \cdots&\cdots&\cdots& \cdots
\\ \hhline{~----}
\end{tabu}$
}
\end{center}

\vspace{-1mm}
Tuple $\tau_2$ in $R^\star$ gives rise to  a block, $\mc{B}(\tau_2)$, of five tuples, $\tau_2^1, \ldots, \tau_2^5$, with
 probabilities $p_2^1:=p^{\nit{BID}}(\tau_2^1), \ldots, p_2^5:= p^{\nit{BID}}(\tau_2^5)$; etc.  In this way, we obtain a BID as  in Figure \ref{tab:dbex0}(b).
\ In  block
$\mc{B}(\tau_2)$, the probability of tuple $\tau_2^1$ is defined on the basis of (or conditioned to)  the observed values in the same tuple (similarly for the other tuples):
{\footnotesize \begin{equation}p^{\nit{BID}}(\tau_2^1) := P^\mc{M}(C^m=c_1~|~A^o=a_2,B^o=1,C^\star=\na). \label{eq:bid-tuple}
\end{equation}}

By choosing one tuple per block, the BID in Table \ref{tab:blocks}(b)
gives rise to a set $\mc{W}(R^\star)$ of possible worlds, $W$, among them, those in Table \ref{tab:two},  where the underlined values are obtained by replacing {\footnotesize \na} \ by domain values.
Their probabilities are: \  $P^{\mc{W}}(R_1) = 1\times p_2^1 \times 1 \times 1 \times 1 \times p_6^5$, \ and \ $P^{\mc{W}}(R_2) = 1\times p_2^2 \times 1 \times 1 \times 1 \times  p_6^3$, resp. \boxtheorem
 \end{example}

\vspace*{-5mm}
\captionof{table}{Two Possible Worlds.}\label{tab:two}\vspace{-2mm}
\begin{center} {\tiny
$\begin{tabu}{c|c|c|c|}\hline
R_1 & A^o & B^o & C^m\\ \hline
\tau_1& a_1 & 0& c_1\\
\tau_2^1& a_2& 1& \underline{c_1}\\
\tau_3&a_3 & 0&c_3\\
\tau_4& a_4&0&c_4\\
\tau_5& a_5&1&c_5\\
\tau_6^5&a_6&1& \underline{c_5}\\ \hhline{~---}
\end{tabu}$
~~~~~
$\begin{tabu}{c|c|c|c|}\hline
R_2 & A^o & B^o & C^m\\ \hline
\tau_1& a_1 & 0& c_1\\
\tau_2^2& a_2& 1& \underline{c_2}\\
\tau_3&a_3 & 0&c_3\\
\tau_4& a_4&0&c_4\\
\tau_5& a_5&1&c_5\\
\tau_6^3&a_6&1& \underline{c_3}\\ \hhline{~---}
\end{tabu}$
}
\end{center}

\ignore{++
\begin{definition}
\label{def:blocks}
Consider an MG $\mc{M}$ and  an observed instance $D^\star$ for schema $S^\star$, and $R^\star$ a relation in $D^\star$ with schema $R(\bar{A}^o,\bar{A}^\star)$, where $\bar{A}^o, \bar{A}^\star$ are lists of fully observed and possibly taking \na \ attributes, resp.  Let $\tau$ be a tuple in $R^\star$, and $\tau[\bar{A}^o,\bar{A^\prime}^\star]$ its restriction to those attributes without an {\footnotesize \na}, with $\bar{A^\prime}^\star \subseteq \bar{A}^\star$.

\vspace{1mm}\noindent
(a) The {\em block} associated to $\tau$ is the set of tuples for schema $R(\bar{A}^o,\bar{A}^m)$: \
$\mc{B}(\tau) := \{  \tau^\prime~|~\tau^\prime[\bar{A}^o,\bar{A^\prime}^\star] = \tau[\bar{A}^o,\bar{A^\prime}^\star], \mbox{ and},\mbox{ for each } A \in (\bar{A}^\star \smallsetminus  \bar{A^\prime}^\star),  \tau^\prime[A] \in \nit{dom}(A^m)  \}$.

\noindent (b)
For $\tau^\prime \in \mc{B}(\tau)$, its  probability $p(\tau^\prime)$ is the conditional probability:
$p^{\nit{BID}}\!(\tau^\prime) := P^\mc{M}(\tau^\prime[\bar{A}^\star \smallsetminus  \bar{A^\prime}^\star]~|~\tau)$.
The probability of the tuple in a singleton block  is $1$.

\vspace{1mm}\noindent (c)
$D^p\!(D^\star,\mc{M})$ denotes the BID whose relations $R^p$ contain the blocks $\mc{B}(\tau)$ for $\tau \in R^\star$, and each tuple $\tau^\prime \in R^p$ has probability $p^{\nit{BID}}(\tau^\prime)$.

\vspace{1mm}\noindent (d)
A {\em possible world} associated to $D^\star$ is an instance $W$ for the underlying  schema $S$, with relations $R^W$ that contain, for each $\tau \in R^\star$, only one $\tau^\prime \in \mc{B}(\tau)$, and nothing more. \ $\mc{W}(D^\star)$ denotes the set of possible worlds. \boxtheorem
\end{definition}++}

  A {\em possible world} associated to $D^\star$ is an instance $W$ for   schema $S$, with relations $R^W$ that contain, for each $\tau \in R^\star$, only one $\tau^\prime \in \mc{B}(\tau)$, and nothing more.
 Possible worlds  may contain duplicates which we tell apart via tids. In Example \ref{ex:first}, duplicates are generated from $\tau_2$ and $\tau_6$ by replacing \na \ by $c_3$ in the former, and by $a_2$ and  $c_3$ in the latter.

\vspace{-1mm}
\begin{example}\label{ex:connection}  \ (Ex. \ref{ex:new2} cont.) \ It is easy to check that, starting from (\ref{eq:bid-tuple}), the  probabilities $p_2^i$ of the tuples $\tau_2^i$,  \ $i=1,\ldots,5$, in block $\mc{B}(\tau_2)$ in Table \ref{tab:blocks}(b) are as follows:

\vspace{-5mm}
{\footnotesize
\begin{eqnarray*}
    p^\nit{BID}(\tau_2^i)\!\!&=&\!\!\frac{P^\mc{M}(A^o=a_2,B^o=1,C^m=c_i, \mbb{I}^C=1)}{P^\mc{M}(A^o=a_2, B^o=1,\mbb{I}^C=1)}\\ &=& P^\mc{M}(C^m =c_i).
\end{eqnarray*}}

\vspace{-9mm}\boxtheorem
\end{example}

\vspace{-2mm}
Each possible world $W$ of the BID becomes one version of the observed  DB $D^\star$ obtained via multiple imputation \cite{gelman}. \ignore{Accordingly, our approach can be seen as  doing {\em multiple imputation} \cite{allisonMI,gelman}, with all the -possibly implicitly and virtually- generated instances are simultaneously queried, as we show now.}

Since each possible world $W$ of the BID is an instance for the underlying schema $S$, a query $\mc{Q}$ posed to the BID  will be expressed in language of the underlying schema, using only attributes of the forms $A^o$ and $A^m$.
 {\em Then, by definition, {\em posing a query $\mc{Q}$ to the observed database $D^\star$} means querying the associated BID $D^p(D^\star,\mc{M})$.}

It is easy to see that the number of possible worlds for  $D^p(D^\star,\mc{M})$ can be exponential in the size of $D^\star$. Furthermore, the QA problem is bound to be hard in data complexity, because
the problem of computing the probability of BQ on a  TID \cite{suciu} can be reduced in polynomial-time in the data to QA under an observed DB with its MG.

\vspace{-3mm}
    \begin{theorem} \label{thm:sharp} Conjunctive query answering on observed databases with missing values via the generated BIDs is $\#P$-hard. \boxtheorem
         \end{theorem}

  \vspace{-4mm}   Instead of going deeper into the investigation of general QA on the resulting BIDs, we will propose, starting in Section \ref{sec:queries},  an alternative QA semantics. The presence of duplicates will be particularly important.  \ignore{In this work, we will mostly work with a bag semantics.}
\ignore{Beware:We end this section with an important observation.
\vspace{-3mm}
\begin{remark} \label{rem:obs}One could wonder if any BID can be the result from an observed DB with a MG. The answer is positive, and relevant to establish some complexity results later in this work. Due to its auxiliary nature, we relegate details of this result to Appendix \ref{app:charac}. \boxtheorem
\end{remark}
}

\vspace{-1mm}
\section{Class-Based Data \\ Semantics}
\label{sec:queries}

\vspace{-1mm}
Our data QA semantics, that builds upon the classic one for BIDs, leverages the stochastic and statistical origins of the observed DB $D^\star$.  We adopt a two-dimensional perspective: \ (a) Each possible world is associated with a probability that reflects its stochastic uncertainty. This is what we have so far. \  (b) The second dimension is of a statistical nature: A notion of \emph{compliance}  quantifies how well a possible world aligns with the joint distribution induced by the MG. This second dimension is developed in detail in Section \ref{sec:compliance-pw}. In this section, we prepare the ground and bring up some relevant issues by means of our running example.

  \vspace{-2mm}
\captionof{table}{(a) Obs. $D^\star$,  (b) BID $D^p$,  (c) MPD w/prob. $(\frac{1}{2})^3$.}\label{fig:two} \vspace{-1mm}
 \begin{center}
 {\tiny   $\begin{tabu}{c|c|c|l|}
\hhline{~---}
&A^o & B^o & C^\star \\ \hhline{~---}
\tau_1 &a & 0  & 0  \\ \hhline{~---}
\tau_2 &a & 0  & 0  \\ \hhline{~---}
\tau_3&a & 1  & \na \\ \hhline{~---}
\tau_4&a & 1  & 0  \\ \hhline{~---}
\tau_5&a & 1  & \na \\ \hhline{~---}
\tau_6&a & 1  & 1  \\ \hhline{~---}
\tau_7&a & 1  & \na \\ \hhline{~---}
\tau_8&a & 1  & 2  \\ \hhline{~---}
\end{tabu}$}~~~~~~~~{\tiny
\begin{tabular}{l|l|l|l||l|}
\cline{2-5}
              & $A^o$ & $B^o$ & $C^m$ & $p^\nit{BID}$ \\ \hline
$\tau_1$                  & $a$ & 0  & 0  & 1 \\ \hline
$\tau_2$                  & $a$ & 0  & 0  & 1 \\ \hline
\multirow{3}{*}{$\mc{B}(\tau_3)$} & $a$ & 1  & 0  &  $1/2$ \\ \cline{2-5}
                    & $a$ & 1  & 1  & $1/4$   \\ \cline{2-5}
                    & $a$ & 1  & 2  & $1/4$  \\ \hline
$\tau_4$                  & $a$ & 1  & 0  & 1 \\ \hline
\multirow{3}{*}{$\mc{B}(\tau_5)$} & $a$ & 1  & 0  & $1/2$  \\ \cline{2-5}
                    & $a$ & 1  & 1  &  $1/4$ \\ \cline{2-5}
                    & $a$ & 1  & 2  & $1/4$  \\ \hline
$\tau_6$                  & $a$ & 1  & 1  & 1 \\ \hline
\multirow{3}{*}{$\mc{B}(\tau_7)$} & $a$ & 1  & 0  & $1/2$  \\ \cline{2-5}
                    & $a$ & 1  & 1  &  $1/4$ \\ \cline{2-5}
                    & $a$ & 1  & 2  &  $1/4$ \\ \hline
$\tau_8$                  & $a$ & 1  & 2  & 1 \\ \hline
\end{tabular}}

\vspace{1mm}
{\tiny $\begin{tabu}{c|c|c|c|}
\hhline{~---}
&A^o & B^o & C^m \\ \hhline{~---}
\tau_1&a & 0  & 0 \\ \hhline{~---}
\tau_2&a & 0  & 0 \\ \hhline{~---}
\tau_3^1&a & 1  & 0 \\ \hhline{~---}
\tau_4&a & 1  & 0 \\ \hhline{~---}
\tau_5^1&a & 1  & 0 \\ \hhline{~---}
\tau_6&a & 1  & 1 \\ \hhline{~---}
\tau_7^1&a & 1  & 0 \\ \hhline{~---}
\tau_8&a & 1  & 2 \\ \hhline{~---}
\end{tabu}$}
\end{center}

\vspace{-2mm}
\begin{example}  \label{ex:aggre} (Ex. \ref{ex:connection}  cont.)
  Consider the MG in Example \ref{ex:second}, the observed DB in Table \ref{fig:two}(a), and $\nit{dom}(C^m) = \{0,1,2\}$. \
   Assume: $P^{\mc{M}}(C^m = 0)=\frac{1}{2}$, $P^{\mc{M}}(C^m =1)= P^{\mc{M}}(C^m = 2) =\frac{1}{4}$, that, by Example \ref{ex:connection},  is all we need to build the BID in Table \ref{fig:two}(b). Table \ref{fig:two}(c) shows the most-probable possible world, $W_1$ in Table \ref{tab:ex-worlds}. It has duplicates.

   Table
\ref{tab:ex-worlds} shows all possible worlds. We use the  notation $W^v$ for possible worlds where $v$ is a vector of values from $C^m$'s domain of length $n$, the number of (non-singleton) blocks in the BID. Here, $n=3$. World $W^{[0,0,0]}$ is obtained by choosing $0$ for the MVs for $C^m$ in each of the three blocks of the BID. Similarly, $W^{[1,2,0]}$ is obtained by choosing $1$ for block 1; $2$ for  block 2; and $0$ for block 3. The last column of Table \ref{tab:ex-worlds} shows the probability $P^\mc{W}(W)$ of  each world $W$. \boxtheorem
\end{example}

\vspace{-7mm}
 \begin{center}
 \captionof{table}{Possible Worlds with Probabilities.}
     \label{tab:ex-worlds}
\tiny{
$
 \begin{array}{|c|c|l|}
 \hline
  \mathbf{World} &
 \mathbf{Notation} &  ~~~\mathbf{P^\mc{W}(W)}   \\ \hline
W_1 & W^{[000]} &  (1/2)^3=\mbf{0.125} \\ \hline
W_2 & W^{[001]} &  (1/2)^2\times 1/4=0.06  \\ \hline
W_3 &  W^{[002]}  &(1/2)^2\times 1/4=0.06  \\ \hline
W_4 &  W^{[010]}  & (1/2)^2\times 1/4=0.06 	  \\ \hline
 %
 %
\ldots &  \ldots  & \ldots  \\ \hline
 %
 %
 %
 %
 %
 %
W_{20} & W^{[201]} & (1/4)^2\times 1/2=0.03 \\ \hline
W_{21} &  W^{[202]}  & (1/4)^2\times 1/2=0.03 \\ \hline
W_{22}  & W^{[210]}  &(1/4)^2\times 1/2=0.03 \\ \hline
\ldots &  \ldots  & \ldots  \\ \hline
%
 %
W_{26}  & W^{[221]}  &(1/4)^3=0.015 	 \\ \hline
W_{27} & W^{[222]}  & (1/4)^3=0.015   \\ \hline
 \end{array}$}
\end{center}

\ignore{
\vspace{-2mm}
\subsection{A Class-Based Data and QA Semantics}\label{sec:classes}
}

Our running example shows that some  possible worlds coincide as multisets except for the tuples's tids. This is the case of $W^{[002]}$ and $W^{[020]}$. We can also obtain the same answer to a query from different possible worlds (see Example \ref{ex:aggre+}). \ignore{shows that the answer to the query $\nit{sum}(C^m)$ from a world $W^{[x_1,x_2,x_3]}$ is of the form $3 + x_1 + x_2 + x_3$, with $x_1,x_2,x_3 \in \{0,1,2\}$, chosen for the blocks $\mc{B}(\tau_3), \mc{B}(\tau_5)$ and $\mc{B}(\tau_7)$. Worlds $W^{[002]}$ and $W^{[110]}$, that make different imputation choices and do not coincide as multisets, produce the same aggregate value. All the worlds in a class return the same answer to a query.}

Given the {\em accidental} coincidence of worlds due to the implicit imputation process, we
will group together possible worlds into {\em classes of worlds}, each class containing the worlds that coincide as multisets modulo the tids. Clearly, worlds in a same class return the same answer to a query. \ignore{++ equivalence classes, on the basis of  their tuple contents as multisets or bags of tuples.  \ Definition \ref{def:matching} will bring into a same class any two worlds that contain exactly the same tuples, regardless of how they were derived. ++}

\vspace{-2mm}
\begin{definition} \label{def:matching} (a) Possible worlds $W,W^\prime \in \mc{W}(D^\star,\mc{M})$ obtained from a BID.  are said to be {\em matching} iff they become the same multiset instance when stripped from tids. \
(b) A {\em class of possible worlds}, $C$, is a maximal subset (under set-inclusion) of $\mc{W}$ where all  worlds in it are matching with each other. \  $\mc{C}(D^\star,\mc{M})$ denotes the collection of all classes. \
(c) The probability of a class $C$ is: \ $
P^\mc{C}\!(C) \ := \  \sum_{W \in C} P^\mc{W}\!(W)$. \
(d) $C$ is a {\em most probable class} (MP-class) if $P^\mc{C}(C)$ takes a maximum value in $\mc{C}(D^\star,\mc{M})$. \ $\nit{MP}(\mc{C}(D^\star,\mc{M}))$ denotes the collection of MP-classes.  \ (e) Given a query $\mc{Q}$ (in the in the language of  schema $S$), and a class $C \in \mc{C}(D^\star,\mc{M})$, the answer to $\mc{Q}$ from $C$ is \ $\mc{Q}[C] := \mc{Q}[W]$, with arbitrary $W \in C$, and its probability is $P^\mc{C}(C)$.
 \boxtheorem
\end{definition}

\vspace{-5mm}
\begin{example} \label{ex:aggre+}  (Ex. \ref{ex:aggre} cont.) Table \ref{tab:classes2prime} shows the different classes obtained from the worlds in Table~\ref{tab:ex-worlds};  and their probabilities. With the aggregate query $\mc{Q}\!\!: \! \nit{sum}(C^m)$ over the partially observed attribute $C^m$, we obtain the  answers in Table \ref{tab:classes2prime}.
\ignore{$
\{\langle \nit{sum} \! = \! 3; \ 0.125 \rangle,\ \langle \nit{sum} \! = \! 4; \ 0.1875 \rangle,\ \langle \nit{sum} \! = \! 5; \ 0.28125 \rangle,\ \langle \nit{sum} \! = \! 6; \ 0.203125 \rangle,\ \langle \nit{sum} \! = \! 7; \ 0.140625 \rangle,\ \langle \nit{sum} \! = \! 8; \ 0.046875 \rangle,\ \langle \nit{sum} \! = \! 9; \ 0.015625 \rangle\}$.} The probability of an answer  is the sum of the probabilities of the classes that return it.
 Answer  $\nit{sum} = 5$ is the most probable answer, with probability $0.28125$.
\boxtheorem
\end{example}

\vspace{-8mm}
\begin{center}
\captionof{table}{Classes and Answers  (bag semantics).}\label{tab:classes2prime}
\scriptsize{
 $\begin{tabu}{|c|c|c|c|c|}
 \hline
  \text{Classes}     & \text{Worlds}  & P^\mc{C}(C) &   \nit{sum}(C^m)[C]   \\ \hline
C_1 & W^{[002]},  W^{[020]},   W^{[200]} 	& 0.188 & 5 \\ \hline
C_2 & W^{[011]},  W^{[101]}, W^{[110]}  &	0.094  & 	 5 \\ \hline
C_3 &  	W^{[012]}, W^{[021]}, W^{[102]},  &0.188&  6\\
&W^{[120]},W^{[201]},W^{[210]}& &
\\ \hline
C_4 & W^{[111]} 	&0.016 &  6 \\ \hline
 C_5 & W^{[001]}, W^{[010]}, W^{[100]} &0.188	 & 4	\\ \hline
C_6 & W^{[022]}, W^{[202]}, W^{[220]} & 0.094	  & 7	\\ \hline
C_7 & W^{[112]}, W^{[121]}, W^{[211]}  & 0.047	  & 7	\\ \hline
C_8 & W^{[000]} 	 & 0.125   & 3\\ \hline
C_9 & W^{[122]}, W^{[212]}, W^{[221]} & 0.047  & 8	\\ \hline
C_{10} & W^{[222]}  & 0.016 4 & 9	 \\ \hline
 \end{tabu}$}
\end{center}

\ignore{
\vspace{-3mm}
\begin{example} \label{ex:sum} (ex. \ref{ex:aggre+} cont.)
  Table \ref{tab:classes2prime} shows the classes obtained from  the worlds in Table \ref{tab:ex-worlds}.
  \ The query answers  from the  classes are in the last column of Table \ref{tab:classes2prime}. All worlds in the same class yield the same answer for every query, under the set or bag semantics.
  \  Classes $C_1$ and $C_2$ have different probabilities. Worlds in a same class may have different global probabilities.
\boxtheorem
\end{example}
}

\vspace{-2mm}
\begin{remark} \label{rem:canon} (canonical representation of classes)
Given the BID $D^p(D^\star,\mc{M})$, with a set of $n$ blocks $\mc{B} = \{\mc{B}_1, \ldots, \mc{B}_n\}$ (including singletons), a class $C$ of worlds is determined by the multiplicities of the tuples they contain.  Let $T = \langle t_1,\ldots, t_m\rangle$, called the {\em support} of the BID, be the {\em vector of distinct} tuples appearing in the BID (without considering the tids).  So, $T$ has a fixed enumeration of tuples. (We will still use the notation $t \in T$ and $|T|$.)
Let $n_j$ denote the multiplicity of tuple $t_j$ across the blocks in $\mc{B}$, i.e. its  number of occurrences in the BID. Accordingly, a class $C$ is uniquely determined by a vector of positive integers $\mathbf{k}=\langle k_1, \ldots, k_m\rangle$, with $\sum_{j \in [1,m]} k_j=n$, where each $k_j \leq n_j$ is the number of occurrences of  $t_j$ in $C$. This class is denoted with $C_{\mathbf{k}}$. \ We denote with $\nit{adm}(\mathbf{k})$
 the fact that $\mathbf{k}$ is admissible, that is, there is a class $C_{\mathbf{k}}$ for $D^p(D^\star,\mc{M})$. When {\em searching} for such a $\mathbf{k}$, i.e. searching for a class with good properties, we have to check admissibility.
\boxtheorem
\end{remark}

\vspace{-4mm}
\begin{example} (ex. \ref{ex:aggre} cont.) \label{ex:support} Consider the BID in Table \ref{fig:two}(b). Here, $n = 8$ and $T = \langle(a,0,0),$ $ (a,1,0),$ $  (a,1,1), (a,1,2)\rangle$, in this order, with  $m=4$. For  $t_3 \ = (a,1,1)$,  $n_3 = 4$. Class $C_2$ contains $W^{[011]} = \{(a,0,0),$ $ (a,1,0), (a,0,0), (a,1,0), \ (a,1,1), \ (a,1,1),$ $ (a,1,2), (a,1,1)\}$, and its matching worlds (see Table \ref{tab:classes2prime}). $C_2$  is characterized by the vector $\mathbf{k} = \langle 2,2,3,1\rangle$, and denoted $C_{\langle 2,2,3,1\rangle}$. \boxtheorem
\end{example}

\ignore{\vspace{-4mm}
\begin{remark} \label{rem:admiss} Given a BID $D^p(D^\star,\mc{M})$ as in Remark \ref{rem:canon}, and a vector of integers $\mathbf{k}$, we denote with $\nit{adm}(\mathbf{k})$
 the fact that $\mathbf{k}$ is admissible, that is, there is a class $C_{\mathbf{k}}$ for $D^p$. In particular, $\sum_{j=1}^m k_j = n$.  When {\em searching} for such a $\mathbf{k}$, in essence searching for a class with good properties, we will assume admissibility can be checked (or, equivalently, that there are cardinality constraints in place that enforce it). We expect the BID to be clear from the context.\boxtheorem
\end{remark}}

\vspace{-2mm}
Now, we can consider a set $\mc{C}^\nit{pref} \subseteq \mc{C}$ of {\em preferred classes}, those with a desired property.  For example, we already have the {\em most-probable classes}. In Section \ref{sec:compliance-pw}, we will consider  {\em most-compliant classes}.  QA can be defined in general, on  an arbitrary set of preferred classes.

\vspace{-2mm}
\begin{definition} \label{def:preferred}  (Preferred QA-Semantics)  Given $\mc{C}^\nit{pref} \subseteq \mc{C}(D^\star,\mc{M})$ and a
relational query $\mc{Q}$: \   The {\em set of all preferred answers}  is:
\ $\nit{Ans}(\mc{Q},\mc{C}^\nit{pref}) \ := \ \{\langle \mc{Q}[C], P^\mc{C}(C)\rangle \ | \ C \in \mc{C}^\nit{pref\!}\}$.
\boxtheorem
\end{definition}

\vspace{-2mm}
Several computational problems arise in relation to this general formulation of QA, and the notion of preferred class.  Some problems will be presented in  Section \ref{sec:complexity}, where we will concentrate mostly on class-related computational problems, leaving QA aside, which is easier than computing classes with certain properties.

\section{Possible-\!World\\ Compliance} \label{sec:compliance-pw}

We can go beyond the purely probabilistic dimension by introducing a new and natural dimension for data and QA semantics: \textit{world compliance}, which we develop in this section. It quantifies how well a possible world conforms to the joint distribution induced by the MG. Actually, in the rest of this paper, we will concentrate on {\em world and class compliance}, leaving aside, due to the lack of space, the development of the probabilistic dimension for an extended version of this work.

Some possible worlds (or classes thereof)  may be {\em more compliant} than others w.r.t. the underlying MG $\mc{M}$, with compliance as a measure {\em statistical fidelity}. It can be cast as a {\em distance} between the data distribution of a world $W$ (and of those in its class) and the  distribution induced by $\mc{M}$.
\ In our running example,
some of the worlds that contribute to the most probable answer, $\nit{sum} = 5$, such as  $W^{[110]}$ and  $W^{[101]}$, may not be the most  compliant. \ignore{Later in this Section, we will define, as particular cases of Definition \ref{def:preferred}, some {\em preferred classes}, $\mc{C}^\nit{mc} \subseteq \mc{C}(D^\star,\mc{M})$,  that contain worlds that are most compliant.}

 In order to define compliance,  we start with the {\em empirical distribution} of a world (which is the same for all the worlds in its class). It is based on the multiset nature of worlds.

\vspace{-1mm}
 \begin{definition} \label{def:emp}
 Given a BID $D^p(D^\star,\mc{M})$, and its {\em support} $T$ (see Remark \ref{rem:canon}), the {\em empirical distribution}, $P_{W}^{\!E}$, of a possible world $W \in \mc{W}(D^\star,\mc{M})$ as a multiset is, for $t \in T$: $P_{W}^{E}(t) \! := \! \nit{mult}^W\!(t)/||W||$,
with $\nit{mult}^W\!(t)$ the multiplicity of $t$ in $W$, and $||W||$ is the bag-cardinality of $W$, i.e. counting duplicates.  \boxtheorem
\end{definition}

\vspace{-2mm}
Notice that the support $T$ does not have duplicates. A tuple in it that does not belong to a particular $W$ has empirical probability $0$ in that world.
The empirical distribution  will be compared with  $P^\mc{M}$, the distribution {\em induced} by the MG $\mc{M}$, as a marginal for the underlying schema $S$.  The comparison is made considering only tuples in $T$.

\begin{example}\label{ex:emp} (ex. \ref{ex:support} cont.) Consider  world $W^{[002]}$ in Table \ref{tab:distr} belonging to class $C_1$ in Table \ref{tab:classes2prime}; and its associated duplicate-free world, denoted with $W^{[002]}\!\!\downarrow$, with its and empirical distribution.
The two other worlds in  $C_1$ share the same empirical distribution. \ignore{Since all the tuples shown in $D^p$ belong to this world, the support is $T = \{t_1,t_2,t_3,t_4\}$.}

\captionof{table}{A World's Associated Distributions.}\label{tab:distr} \vspace{-2mm}
\begin{center}{\scriptsize
\begin{tabular}{l|l|l|l|}
\cline{2-4}
        $W^{[002]}$      & $A^o$ & $B^o$ & $C^m$ \\ \hline
$\tau_1$                  & $a$ & 0  & 0  \\ \hline
$\tau_2$                  & $a$ & 0  & 0  \\ \hline
$\mc{B}(\tau_3)$ & $a$ & 1  & 0  \\  \hline
$\tau_4$                  & $a$ & 1  & 0 \\ \hline
$\mc{B}(\tau_5)$ & $a$ & 1  & 0   \\ \hline
$\tau_6$                  & $a$ & 1  & 1  \\ \hline
$\mc{B}(\tau_7)$
                    & $a$ & 1  & 2  \\ \hline
$\tau_8$                  & $a$ & 1  & 2   \\ \hline
\end{tabular}}
\end{center}

\vspace{-3mm}
\begin{center}
{\scriptsize
\begin{tabular}{l|l|l|l||c|c|}
\cline{2-6}
        $W^{[002]}\!\!\downarrow$     &\!\!$A^o$\!\!&\!\!$B^o$\!\!&\!\!$C^m$\!\!&\!\!$P^E_{\!W[002]}$\!\!&\!\!$P^\mc{M}$\!\!\\ \hline
$t_1$                  & $a$ & 0  & 0 & 1/4 &0.225\\ \hline
$t_2$ & $a$ & 1  & 0  & 3/8& 0.225\\   \hline
$t_3$                  & $a$ & 1  & 1  & 1/8 &0.1125\\ \hline
$t_4$
                    & $a$ & 1  & 2 & 1/4 & 0.1125\\ \hline
\end{tabular}}
\end{center}

Let us assume, consistently with Example \ref{ex:aggre}, that $\nit{dom}(C^m) = \{0,1,2\}, \nit{dom}(A^o) = \{a,b\}, \nit{dom}(B^o) = \{0,1\}, \nit{dom}(C^\star) = \{0,1,2,\na\}$, and also: {\footnotesize $
P^\mc{M}(C^m=0)=1/2,
P^\mc{M}(C^m=1) = P^\mc{M}(C^m=2) = 1/4,
P^\mc{M}(B^o=0) = P^\mc{M}(B^o=1) = 1/2,
P^\mc{M}(A^o=a)=0.9,  \
P^\mc{M}(A^o=b) = 0.1$}.
 As noted in Example \ref{ex:aggre}, with these probabilities we can compute the induced probabilities (using equations (\ref{eq:ass1})-
(\ref{eq:ass5}), (\ref{eq:joint})): \  For $x \in \{a,b\}, y \in \{0,1\}, z \in \{0,1,2\}$:\ignore{\footnote{Details in Example \ref{ex:emp2} in the Appendix.}}
{\footnotesize $P^\mc{M}(A^o\!=\!x,B^o\!=\!y,C^m\!=\!z)= \sum\limits_{u \in \{0,1\}, v \in \{0,1,2,\na\}} \hspace{-8mm}P^\mc{M}(A^o\!=\!x,B^o\!=\!y,C^m\!=\!z,\mbb{I}^C\!=\!u, C^\star\!=\!v)
=
P^\mc{M}(C^m\!=\!z) \!   \times \! P^\mc{M}(B^o\!=\!y) \! \times \! P^\mc{M}(A^o\!=\!x)$}.
\ignore{ Accordingly, {\footnotesize $P^\mc{M}(a,0,0) = P^\mc{M}(a,1,0) = 0.9 \times 1/2 \times 1/2= 0.225$,  $P^\mc{M}(a,1,1) = P^\mc{M}(a,1,2) = 0.9 \times 1/2 \times 1/4= 0.1125$}.} We obtain the last column, $P^\mc{M}$, in Table \ref{tab:distr}, also shared with the other worlds in class $C_1$.
\boxtheorem
\end{example}

\vspace{-4mm}
\paragraph{Distance-Based Compliance.}

If  we have  an abstract measure of {\em distance}, $d(\cdot,\cdot)$,  between an empirical distribution   and  the joint distribution $P^\mc{M}$, we can define compliance.

\vspace{-2mm}
\begin{definition} \label{def:dist}
 (a) The {\em  compliance degree} of $C \in \mc{C}(D^\star,\mc{M})$ is: \
    $d^c(C, \mc{M})  \ := \ d(P^E_W,P^\mc{M}), \ \mbox{ with any } W \in C$. \
(b) $C$ is a {\em most-compliant class} (an MC-class) w.r.t. $d$ if $d^c(C, \mc{M})$ takes a minimum value in $\mc{C}(D^\star,\mc{M})$.  $\mc{C}^\nit{mc}_{\!d}$ denotes the collection of MC-classes relative to $d$.
 \boxtheorem
\end{definition}

\vspace{-2mm}The distance in (a) is well defined since all worlds in a class of matching worlds have the same empirical distribution.

Several distances between probability distributions offer themselves to define the compliance degree, among them the {\em Kullback-Leibler Divergence} (KLD) \cite{larry}:

\vspace{-4mm}
{\scriptsize \begin{align}
d^{\mathrm{KL}\!}(P^E_W, P^\mc{M}) &:= \mathrm{Div_{KL}}(P_W^E \parallel P^\mc{M})
:= \sum_{t \in T} P_W^E(t) \, \ln \frac{P_W^E(t)}{P^\mc{M}(t)}, \label{eq:kld}
\end{align}}

\vspace{-1mm} so as the {\em Euclidean Distance}, $d^{\nit{EU}}$ \cite{hastie}.  For a class $C$, $\nit{KLD}(C)$ denotes the KL-distance in common to all worlds in $C$ to $P^\mc{M}$.
Any two matching worlds, $W, W^\prime$, become equally compliant. \
However, they may have different global probabilities, $p^\mc{W}(W)$ and $p^\mc{W}(W^\prime)$.

\vspace{-2mm}
\begin{example} \label{ex:sum+} (ex. \ref{ex:emp} cont.) With (\ref{eq:kld}), we can compute the KL-distance for each class. \ignore{for world $W^{[002]}$, which is the same for all the worlds in its class $C_1$:
$d^{\mathrm{KL}\!}(P^E_{W^{[002]}}, P^\mc{M}) =  \sum\limits_{j=1}^{4} P_{W^{[002]}}^E(t_j) \, \ln \frac{P_{W^{[002]}}^E(t_j)}{P^\mc{M}(t_j)}
= 1/4 \times \ln(\frac{1/4}{0.225}) + 3/8 \times \ln(\frac{3/8}{0.225}) +  1/8 \times \ln(\frac{1/8}{0.1125}) + 1/4 \times \ln(\frac{1/4}{0.1125})
= 0.026 + 0.192  + 0.013 + 0.200 = 0.431$.}
Table \ref{tab:classes2}  shows the distances of the different classes to the induced distribution.
 The most-compliant classes (MC-classes) are $C_1, C_5, C_8$. The most probable classes are $C_1, C_3, C_5$.
\boxtheorem
\end{example}

\vspace{-2mm}Consistently with Definition \ref{def:dist}(b), we denote with $\mc{C}^\nit{mc}_\nit{KL}(D^\star,\mc{M})$ and $\mc{C}^\nit{mc}_\nit{EU}(D^\star,\mc{M})$ the set of MC-classes with respect to the KL- and Euclidean distances, resp.

\vspace*{-5mm}
\begin{center}
\captionof{table}{Classes and Compliance  (bag semantics). }\label{tab:classes2}
\scriptsize{
 $\begin{tabu}{|c|c|c|c|}
 \hline
  \text{Classes} \  C   & \text{Worlds}  & P^\mc{C}(C) &  \nit{KLD}(C)   \\ \hline
C_1 & W^{[002]},  W^{[020]},   W^{[200]} 	& 0.188 & 0.431\\ \hline
C_2 & W^{[011]},  W^{[101]}, W^{[110]}  &	0.094 &0.518\\ \hline
C_3 &  	W^{[012]}, W^{[021]}, W^{[102]},
 &0.188&    0.452\\
&W^{[120]},W^{[201]},W^{[210]} & &
\\ \hline
C_4 & W^{[111]} 	&0.016&   0.711\\ \hline
 C_5 & W^{[001]}, W^{[010]}, W^{[100]} &0.188	&  0.431\\ \hline
C_6 & W^{[022]}, W^{[202]}, W^{[220]} & 0.094	&  0.711\\ \hline
C_7 & W^{[112]}, W^{[121]}, W^{[211]}  & 0.047	&  0.929\\ \hline
C_8 & W^{[000]} 	 & 0.125 & 	 0.431 \\ \hline
C_9 & W^{[122]}, W^{[212]}, W^{[221]} & 0.047 &  0.972\\ \hline
C_{10} & W^{[222]}  & 0.016 &  1.111\\ \hline
 \end{tabu}$
 }
\end{center}

\vspace{-1mm}
\paragraph{Goodness-of-Fit Compliance.}

An alternative take  on possible-world compliance is based on {\em hypothesis testing}. Given  $W \in \mc{W}$ as a sample, we {\em test the hypothesis}, $H_0$, that  it  fits the induced distribution $P^\mc{M}$. \ We use the $\chi^2$-statistic:

\vspace{-1mm}
{\scriptsize $$\chi^2(W) := \sum_{i=1}^{m} \frac{(P^E_W(\bar{a_i}) \ - \ P^\mc{M}(\bar{a}_i))^2}{P^\mc{M}(\bar{a}_i)}, \ \ \ \ (m=|T|)$$}

\vspace{-1mm}which, under $H_0$, has approximately a $\chi^2_{m-1}$-distribution \cite{larry}. It can be seen as a distance, actually a measure of the {\em relative square deviation} of $P^E_W$ from $P^\mc{M}$.

In order to compare worlds, we can use a {\em compliance order} based on the $p$-values for the test:    $p^{\!V\!\!}(W) \ := \ P_{H_0}(\chi^2 \ \geq \  \chi ^2(W))$, that is the probability (under $H_0$) that $\chi^2$ -as a random variable- is at least as contradictory to $H_0$ as the value of $\chi^2(W)$. It is defined by: $
W_1 <^\mc{M}_\nit{pv} W_2$ iff $p^{\!V\!\!}(W_1) \ < \ p^{\!V\!\!}(W_2)$. \ In this way, we have the family $\mc{C}^\nit{mc}_{pv}$ of most-compliant classes based on  the {\em p-value of the $\chi^2$-test}.

\section{Computing Classes\\ and QA}
\label{sec:complexity}

With a general notion of class-based query answering (QA), and the   most-compliant and the most-probable collections of classes,\ignore{\footnote{These semantics provide alternative perspectives on QA, depending on whether the focus is on statistical likelihood or distributional fidelity.}} we can turn to
 computational problems that naturally arise   \ignore{Among them,  given a BID $D^p(D^\star,\mc{M})$, compute a most-compliant (MC) class, or compute a most- probable (MP) class. Given also a
 query $\mc{Q}$,  compute a world $W$ from a MC-class class $C$, and return $\mc{Q}[W]$ and $P^\mc{C}(C)$. Similarly with $C$ a MP-class. We may also want to do this for all MC- or all MP-classes. \
 Diverse computational problems} and to their computational complexities. \ignore{emerge, among them, decision, functional, optimization, enumeration \cite{strozecki2023}, and counting problems \cite{toran91}.}

We will concentrate only on most-compliant-related computational problems.  They  will be formulated and addressed  using  the class representation in Remark \ref{rem:canon}; in particular,
we consider   classes $C_{\blue{\mathbf{k}}}$, and a representative world therein, $W_{\blue{\mathbf{k}}}$, that  can also be seen as encoding the degree of compliance of its class.

\ignore{
\subsection{MC-Related Computational Problems} \label{sec:mccproblems}
}

The  notion of {\em convex-separability} \cite{ahuja1993} associated to distances will become critical.
Intuitively, such distances $d(\cdot,\cdot)$ can be computed by aggregating independent contributions associated with individual tuples, so that the total distance decomposes as a sum of per-tuple convex terms:
$d(C_{\mathbf{k}}, \mathcal{M})=
\sum_{j=1}^m d_j\!\bigl(k_j, P^\mathcal{M}(t_j)\bigr)$, where  each $d_j$ is a convex function capturing the local cost of the assignment made to tuple $t_j$.
The family of convex separable functions encompasses many classical measures of
discrepancy, including $f$-divergences~\cite{ali1966} such as the
KL-divergence, the Hellinger distance , the total variation distance, and the $\chi^2$-statistic,
as well as $L^p$-norms~\cite{royden1988}, for $1 \leq p < \infty$, such as the
$L^1$ Manhattan distance  and the $L^2$- Euclidean
distance.

\ignore{We recall that the support  is  $T=\langle t_1,\ldots, t_m \rangle$.}

\ignore{\begin{definition} (convex-separable distances/divergences) \label{def:sep}
(a) A distance $d(C_{\mathbf{k}}, \mathcal{M})$, for classes $C_\mbf{k}$ with $\mathbf{k}=\langle k_1, \ldots, k_m\rangle$, is called \emph{separable} if it can be expressed as $d(C_{\mathbf{k}}, \mathcal{M})=
\sum_{j=1}^m d_j\!\bigl(k_j, P^\mathcal{M}(t_j)\bigr)$, where $d_j(k_j,P^\mc{M}(t_j))$ is a function that depends on the number of occurrences of tuple $t_j$ in $C_\mbf{k}$ ($t_j$'s contribution to the distance). \ (b)  $d_j(k_j, P^\mathcal{M}(t_j))$ is  {\em discretely convex} (in its first argument) if,  for all integers $k_j \ge 1$,
the discrete {\em marginal contributions} are non-decreasing: \ $d_j(k_j+1, P^\mathcal{M}(t_j)) - d_j(k_j, P^\mathcal{M}(t_j)) \geq d_j(k_j, P^\mathcal{M}(t_j)) - d_j(k_j-1, P^\mathcal{M}(t_j))$. \ $d(C_\mbf{k},\mc{M})$ is called {\em convex-separable} if, in addition, each of the $d_i$ is discretely convex.
\boxtheorem
\end{definition}
\vspace{-2mm}
The second argument of $d_j$ can be any parameter depending on $t_j$. For our application, we chose $P^\mc{M}(t_j)$.
 It is easy to check that $d^\mathrm{KL}$,  $d^\mathrm{EU}$, and the $|chi^2$-statistic are convex-separable, so as many distances used in practice. are convex-separable, so as all the distances, divergences and norms in Table~\ref{table:ex-distances} in the Appendix. Also the }


The following definition formalizes  computational problems of interest related to MCC.

\begin{definition} \label{def:mcc}
(MC-related problems)
Assume $D^p(D^\star,\mc{M})$ is a BID.
We define the following  computational problems:

\vspace{1mm}\noindent
1. \textbf{MCC} (MC-Class): Compute an admissible\ignore{\footnote{That is, $\nit{adm}(\mathbf{k}^\star)$ holds; in particular $\sum_{j=1}^{m} k_j = n$. According to Remark \ref{rem:admiss}, from now on, we leave this condition on class-vectors implicit.}} class-vector $\mathbf{k}^\star \in \mathbb{N}^{m}$  that minimizes the compliance-distance (or maximizes compliance): \
$C_{\mathbf{k}^\star} \ \in \argmin_{\mathbf{k} \in \mathbb{N}^m, \ \sum_{j=1}^{m} k_j = n} d^c(C_{\mathbf{k}}, \mc{M})$.

\vspace{1mm}\noindent 2. \textbf{\#MCC} (Counting MC-Classes): Count the number of MC-classes, i.e. all class-vectors $\mathbf{k}^\star$, such that $C_{\mathbf{k}^\star}$ minimizes

\vspace{1mm}\noindent 3. \textbf{MCC-Enum} (Enumerating MC-Classes): Effectively enumerate all class-vectors $\mathbf{k}^\star$, such that $C_{\mathbf{k}^\star}$ minimizes $\nit{d}^c(C_{\mathbf{k}},\mc{M})$.
\boxtheorem
\end{definition}


Theorem \ref{thm:mcc} provides  computational complexities of MCC-related problems (see \cite{Papadimitriou1994}).

\begin{theorem} \label{thm:mcc}
(MC-related problems)
Assume $D^p(D^\star,\mc{M})$ is a BID with $n$ blocks and a support $T$ with $m$ distinct tuples.
\ For  a {\em convex-separable distance} $d$,
we have:

\vspace{1mm}\noindent
1.   \textbf{MCC}  is in \textbf{FP}, the class of functional problems computable in deterministic polynomial time.

\vspace{1mm}\noindent 2.
     \textbf{\#MCC} is \textbf{\#P-complete}, i.e. it belongs and is hard for the class of counting solutions of NP decision problems.

\vspace{1mm}\noindent 3.  \textbf{MCC-Enum} is in \textbf{DelayP}, the class of enumeration problems where the time delay between the output of any two consecutive solutions is polynomial in the input size.
\boxtheorem
\end{theorem}

 \ignore{We recall that the class-vectors mentioned in this definition (and everywhere) are subject to the {\em admissibility condition} (see Remark \ref{rem:admiss}).}
\textbf{MCC} is about computing one good class. Once we have it, QA on it is straightforward. \ignore{Computing the probability of the class (as in Definition \ref{def:preferred} of QA) is considered in Section \ref{sec:mpcproblems}.}
 Theorem \ref{thm:mcc}  tells us that computing a single MC-class can be achieved   in polynomial time for a wide class of distance functions. Moreover, the (possibly exponentially many)  MC-classes can be effectively enumerated with polynomial delay. However, counting the number of MC-classes turns out to be $\#P$-complete.
\ These results are obtained through a  reduction from \textbf{MCC} to computing a ``minimum-cost flow" in bipartite graphs and some related problems.

Note that, due to the convex-separability of the KL-divergence,  eucledian distance and the $\chi^2$-statistic, Theorem \ref{thm:mcc} holds in particular for $\mc{C}_{\nit{KL}}^\nit{mc}$, $\mc{C}_{\nit{EU}}^\nit{mc}$ and $\mc{C}^\nit{mc}_{pv}$. \boxtheorem

\ignore{XXXX
\subsection{MP-Related Computational Problems} \label{sec:mpcproblems}

 As in the previous section, we consider class-related problems as opposed to QA on those classes. Unlike the MC-setting where the intrinsic optimization problems are tractable, MP-problems have a significantly higher complexity, where hardness  arises from the need to optimize a \#P-complete function (class-probability) over an possibly  exponentially large solution space of possible class-vectors. Even evaluating the objective function at a single point requires solving a \#P-complete problem. The following theorem introduces the problems and establishes their computational  complexities.

 \begin{theorem} \label{thm:mcc}
  (MP-Related Problems)
Assume $D^p(D^\star,\mc{M})$ is a BID with $n$ blocks and a support $T$ with $m$ distinct tuples. The following are computational problems related to most-probable classes and their complexities:

\vspace{1mm}\noindent
1. \textbf{Class Probability:} Given a class $C_{\mathbf{k}}$, with $\mathbf{k} \in \mathbb{N}^{m}$, compute its probability $P^{\mathcal{C}\!}(C_{\mathbf{k}})$.

\vspace{1mm} \noindent
\textbf{Class Probability} is \textbf{\#P-complete}.

\vspace{1mm}\noindent 2. \textbf{MPC[D]} (Most-Probable Class-Decision Problem): Given a threshold $\theta \in [0,1]$, decide whether there exists a class-vector $\mathbf{k}$\ignore{ = (k_1, \ldots, k_m) \in \mathbb{N}^m$, such that \ \
$\sum_{i=1}^{m} k_i = n
\quad \text{and} \quad  }, such that
$P^\mathcal{C}(C_{\mathbf{k}}) \geq \theta$.

\vspace{1mm}\noindent
\textbf{MPC[D]} is in \textbf{NP\textsuperscript{PP}}.

\vspace{1mm}\noindent 3. \textbf{CMPC[D]} (Constrained MP-Decision Problem): Given a threshold $\theta \in [0,1]$ and an upper bound vector $\mathbf{b}=(b_1, \ldots, b_m) \in \mathbb{N}^m$, decide if there is \ignore{whether there exists} a class-vector $\mathbf{k}$, \ignore{$ = (k_1, \ldots, k_m) \in \mathbb{N}^m$,} such that, for all $j \in [1,m]$,  $k_j \leq b_j$, and $
P^\mathcal{C}(C_{\mathbf{k}}) \geq \theta
$.

\vspace{1mm}\noindent
\textbf{CMPC[D]} is \textbf{NP\textsuperscript{PP}}-complete.

\vspace{1mm}\noindent
4. \textbf{MPC[O]} (Most-Probable Class–Optimization Problem): Find a class-vector $\mathbf{k}^\star$, such that \ignore{\in \mathbb{N}^{m}$ such that
\vspace{1mm}
\hspace*{2.5cm}$C_{\mathbf{k}^\star} \ \in \argmax_{\mathbf{k} \in \mathbb{N}^m, \ \sum_{i=1}^{m} k_i = n}} $P^{\mathcal{C}\!}(C_{\mathbf{k}})$ takes a maximum value.

\vspace{1mm}\noindent \textbf{MPC[O]} is in \textbf{FP\textsuperscript{NP\textsuperscript{PP}}}.
\boxtheorem
\end{theorem}

The $\#P$-completeness of computing a single class probability establishes a fundamental bottleneck: any algorithm for MP-problems must solve this hard counting problem repeatedly during the search process. As a consequence, the decision problems \textbf{MPC[D]} and \textbf{CMPC[D]} lie in $\mathsf{NP^{PP}}$, the class of decision problems that can be solved by a non-deterministic polynomial-time  with access to an oracle for PP (probabilistic polynomial time).

While the lower bound of \textbf{MPC[D]} remains open, the slightly more expressive variant, \textbf{CMPC[D]}, that enables control over tuple multiplicities, is $\mathsf{NP^{PP}}$-complete. The optimization problem \textbf{MPC[O]} sits at the top of this hierarchy in $\mathsf{FP^{NP^{PP}}}$, the class of function problems computable in deterministic polynomial-time with access to an $\mathsf{NP^{PP}}$ oracle.
XXXX}

\section{Related Work}\label{sec:relWork}

Incomplete DBs, in particular in relation to MVs have been investigated for a long time, under different representations and semantics \cite{imielinski,reiter,greco}. None of those approaches is directly our line of work.

Since the inception of relational DBs, null values, in the form of the SQL constant \Nn, have been used to represent MVs. Its use has been contentious and a subject of several papers. A recent body of investigation \ignore{led by L. Libkin and collaborators} has shed light on the  semantic and algorithmic issues related to the use of \Nn. See \cite{libkin} for an interesting discussion and references. For a simple and practical reconstruction of the use of \Nn \ in SQL that is based on  query-rewriting  see \cite{p2p}.
\ More recently, \cite{console} have considered numerical queries in the presence of null values. See also \cite{lastLibkin} for recent results along this line of work. However, following this common approach, i.e, replacing MVs in $D^\star$ by \Nn, and using any SQL-based DB management system for QA, we would be assigning a very particular ``semantics" to data with MVs; one that is embedded in the way SQL DBs {\em operates} with  \Nn.  In our work, we do not refer to- or handle {\em null values} or \Nn~ as in SQL DBs.

In PDBs, a query semantics defines how answers  are interpreted and computed under uncertainty. The \emph{possible worlds semantics} determines probability distributions over those worlds and  over query answers, by evaluating the query on each world \cite{suciu}. Alternatively, the \emph{confidence} or \emph{marginal semantics} focuses on marginal probabilities of individual tuples appearing in results, simplifying computations by ignoring correlations \cite{dalvi2004,soliman2007topk,re2005trio}. Extensions such as the \emph{top-k} and the \emph{expected score} semantics rank or score tuples based on probabilities or utilities \cite{soliman2007topk,sarma2006working,re07}. One can also decide to use the  \emph{most probable database}  to evaluate the query \cite{NOW}. \ For aggregate queries, similar semantics apply: the \emph{possible worlds semantics} yields a distribution over aggregate values \cite{dalvi2004}. \ignore{, \emph{marginal semantics} computes marginal probabilities for each aggregate result, and \emph{expected aggregate semantics} returns the expected aggregate value, also with group-by } \ignore{Advanced semantics also consider ranking or selecting likely aggregates via \emph{top-k} or \emph{most probable aggregate} methods \cite{soliman2007topk,soliman2008,re07}.}

\ignore{
\red{Several query semantics have been proposed for probabilistic databases, each providing a different way to interpret and compute query results: the \textit{possible worlds semantics} models the database as a distribution over deterministic worlds and derives a probability distribution over query answers by evaluating each world \cite{dalvi2004}, \emph{confidence} or \emph{marginal semantics} focuses on marginal probabilities of individual tuples appearing in results, simplifying computations by ignoring correlations \cite{dalvi2004,soliman2007topk,re2005trio}. Extensions such as \emph{top-k} and \emph{expected score} semantics rank or score tuples based on probabilities or utilities \cite{soliman2007topk,sarma2006working,re07}, while \emph{most probable world} semantics selects the single most likely deterministic instance satisfying the query \cite{mostProbDB}.
} }

\ignore{
Various optimization techniques - such as lineage factorization and knowledge compilation- have been developed to enable efficient query evaluation over probabilistic databases \cite{suciu,Olteanu2015}.}

There is a large body of research on dealing with MVs, mostly in Statistics \cite{allisonMI,gelman}. A common technique is {\em imputation}, which amounts to filling in for them using other values in the domain.
\ MMs were proposed in that context \cite{rubinBook}.
Imputation methods and statistical estimates in the presence of MV come in different forms, and may   depend on MMs \cite{rubinBook,gelman}. \  Our work is not about these forms of  ``classic imputation", but, instead and in general terms, we do something like an {\em implicit  causality-informed probabilistic imputation} that  gives rise to several possible ``imputed" instances with attached probabilities.

MMs represented as causal networks \cite{pearl} have been introduced and investigated by Mohan and Pearl \cite{mohanthesis2017,mohan21}. We build upon them. \ignore{developing  a {\em qualitative approach}, by using the data at hand in combination with the MG (as opposed to pre-given distributions) to properly estimate (or {\em recover}) probabilities and statistical aggregations.} \ignore{This treatment of MVs has been applied in machine learning \cite{guy2,broeck2015}. }

In  \cite{benny}, the authors take a probabilistic and general approach to  data quality, which may in principle involve different dimensions of quality, including incompleteness. MVs, as those we deal with, are not specifically  addressed. \ignore{
In \cite{salimi19,salimi24}, techniques based on database repairs w.r.t. multi-valued dependencies (MVDs) are used to make instances comply with stochastic conditional dependencies, but not specifically related to MVs. Similar techniques  could be applied in our case, to make the given instance compliant with the MG. Although we chose to concentrate on the already compliant generated instances,  it would be interesting to explore that direction.}
\ Various methodologies have been proposed to assess the compatibility of a dataset with a Bayesian network,  such as {\em model-fit measures}, e.g.  maximum-likelihood \citep{Bishop2006,allisonML,larry},  {\em information-theoretic measures} \cite{MacKay2003}, and representation of stochastic dependencies via generalized {\em multi-valued dependencies} \cite{wong1}.

\section{Conclusions}
\label{sec:conclusion}

In this work, we have provided a principled semantics to a DB  with MVs. Their  occurrences as assumed to be governed by a quantitative Bayesian-Network,  that represents missingness mechanisms.
\  The data semantics relies on  the construction of a BID that induces a space of {\em possible worlds} with associated probabilities. Possible worlds are multiset-DBs without  MVs.

 To support  QA,  possible worlds are classified in  classes. We introduced and investigated two collections of classes: the \emph{Most Compliant Classes},  the MCC-semantics, which prioritizes classes that best align with the underlying MG; and the \emph{Most Probable Classes}, the MPC-semantics, which selects classes according to their aggregate probability mass. \ They reflect the complementary dimensions of statistical plausibility and probabilistic likelihood.

We presented  complexity results for QA under MCC semantics. Notably, we showed that, assuming access to an oracle for evaluating queries over complete databases, QA under the MCC semantics on a single class is tractable in polynomial time. Although enumerating the answers under all MC-classes is computationally intractable, it can still be performed with polynomial delay.

As part of our ongoing work, not reported here, we have investigated the complexity of the MPC-semantics. With the aim of identifying classes of observed DBs for which efficient query evaluation is tractable, we have uncovered
 a natural class of them for which the
MCC- and the MPC-semantics coincide, and QA becomes tractable.


About future research,
we plan to investigate hybrid query semantics that jointly consider compliance and probability, allowing for a more expressive and flexible notion of QA.

We are also interested in exploring how our class-based semantics could inform or be integrated with advanced imputation strategies, especially in statistical or machine learning pipelines.

We are also interested in the  verification and enforcement of  compliance of an observed DB $D^\star$ with the underlying the MG. Techniques as those reported in  \cite{salimi19,salimi24} could be useful.
Finally, implementing these techniques and conducting experimental evaluations on real-world datasets will be essential for assessing their practical applicability and performance.

\vspace{3mm}
\noindent {\bf Acknowledgements:} \  Leopoldo Bertossi has been financially supported by the ``Millennium Institute for Research on Data"  (IMFD, Chile), and NSERC-DG 2023-04650. Part of this work was done while he was visiting the ``Laboratory of Informatics,
  Modelling and Optimization of the Systems" (LIMOS), at U. Clermont-Ferrand, France. He appreciates their support and hospitality. XXX



\appendix

\section{Appendix}

\begin{definition}
\label{def:blocks}
Consider an MG $\mc{M}$ and  an observed instance $D^\star$ for schema $S^\star$, and $R^\star$ a relation in $D^\star$ with schema $R(\bar{A}^o,\bar{A}^\star)$, where $\bar{A}^o, \bar{A}^\star$ are lists of fully observed and possibly taking \na \ attributes, resp.  Let $\tau$ be a tuple in $R^\star$, and $\tau[\bar{A}^o,\bar{A^\prime}^\star]$ its restriction to those attributes without an {\footnotesize \na}, with $\bar{A^\prime}^\star \subseteq \bar{A}^\star$.

\vspace{1mm}\noindent
(a) The {\em block} associated to $\tau$ is the set of tuples for schema $R(\bar{A}^o,\bar{A}^m)$:
\begin{eqnarray}
\mc{B}(\tau) := \{ \ \tau^\prime~|~\tau^\prime[\bar{A}^o,\bar{A^\prime}^\star] = \tau[\bar{A}^o,\bar{A^\prime}^\star], \mbox{ and},\mbox{ for each }\nonumber \\ A \in (\bar{A}^\star \smallsetminus  \bar{A^\prime}^\star),  \tau^\prime[A] \in \nit{dom}(A^m) \ \}. \hspace{-5mm}\label{eq:block}
\end{eqnarray}
\noindent (b)
For $\tau^\prime \in \mc{B}(\tau)$, its  probability $p(\tau^\prime)$ is the conditional probability:
\begin{equation}
p^{\nit{BID}}\!(\tau^\prime) := P^\mc{M}(\tau^\prime[\bar{A}^\star \smallsetminus  \bar{A^\prime}^\star]~|~\tau).\label{eq:probTuple}
\end{equation}
The probability of the tuple in a singleton block  is $1$.

\vspace{1mm}\noindent (c)
$D^p\!(D^\star,\mc{M})$ denotes the BID whose relations $R^p$ contain the blocks $\mc{B}(\tau)$ for $\tau \in R^\star$, and each tuple $\tau^\prime \in R^p$ has probability $p^{\nit{BID}}(\tau^\prime)$.

\vspace{1mm}\noindent (d)
A {\em possible world} associated to $D^\star$ is an instance $W$ for the underlying  schema $S$, with relations $R^W$ that contain, for each $\tau \in R^\star$, only one $\tau^\prime \in \mc{B}(\tau)$, and nothing more. \ $\mc{W}(D^\star)$ denotes the set of possible worlds. \boxtheorem
\end{definition}

\subsection*{Proof of Theorem \ref{thm:sharp}.}
Theorem \ref{thm:sharp} follows from Proposition \ref{prop:complexity} and the $\#P$-hardness of BCQ evaluation on TIDs \cite{suciu}.

\begin{proposition}\label{prop:complexity} For a fixed underlying schema $S$ and BCQ $\mc{Q}$, there is an observed schema $S^\star$, a qualitative MG $\mc{M}$, and a BCQ $\mc{Q}^\prime$, such that:  For every TID $D$ for $S$, an observed DB $D^\star$ for the observed schema $S^\star$ and a distribution on $\mc{M}$ can be built in PTIME in $|D|$, such that  $\mc{Q}[D] = \mc{Q}^\prime[D^p(D^\star,\mc{M})]$. \boxtheorem
\end{proposition}

Proposition \ref{prop:complexity} above shows that arbitrary TIDs can be recovered as special cases of the kind of BIDs arising from our data semantics. In fact, we establish a more general statement, in Proposition \ref{lemma:BIDrepresentation}: \emph{Any BID can be obtained from a corresponding observed $D^\star$ and an MG $\mc{M}$}. Before going into this result, we show an example that illustrates Proposition \ref{prop:complexity}.

\begin{table}[h]\center
{\footnotesize $\begin{tabu}{c|c|c||c|}\hline
R&A & B& P\\ \hline
\tau_1 & a & b&p_1\\
\tau_2 & a^\prime & b^\prime &p_2\\ \hhline{~---}
\end{tabu}$~$\begin{tabu}{c|c||c|}
\hline
S&A&P\\ \hline
\tau_3 & a&p_3\\
\tau_4 & b&p_4\\
\tau_5 & b^\prime&p_5\\ \hhline{~--}
\end{tabu}$~~~~~$\begin{tabu}{c|c|c|c|}\hline
R^\star&A^o & B^o&M_1^\star\\ \hline
\tau_1 & a & b &\na\\
\tau_2 & a^\prime & b^\prime&\na\\ \hhline{~---}
\end{tabu}$
~$\begin{tabu}{c|c|c|}
\hline
S^\star&A^o&M_2^\star\\ \hline
\tau_3 & a&\na\\
\tau_4 & b&\na\\
\tau_5 & b^\prime&\na\\ \hhline{~--}
\end{tabu}$}
\caption{(a) Initial TID. \ \ \ (b) Associated observed DB. } \label{tab:reduc}
\end{table}

\begin{example} \label{ex:reduc} Consider schema $\mc{S} = \{R(A,B), \ S(B)\}$ for a the TID in Table \ref{tab:reduc}(a). The schema for the associated observed DB with MVs is $\mc{S}^\star = \{R(A^o,B^o,M_1^\star), \ S(B^o,M_2^\star)\}$. \ Attributes $M_1^\star, M_2^\star$ may exhbit MVs, and $\nit{dom}(M_1^m) = \nit{dom}(M_2^m) = \{0,1\}$.

Table \ref{tab:reduc}(a) shows the initial TID that we want to represent as a BID, which we will obtain by first creating the observed instance in Table \ref{tab:reduc}(b). We concentrate on table $R$. \ In order to obtain the BID, which should be the one in Figure \ref{fig:result}(b), we use, for $R^\star$, the MG in Figure \ref{fig:result}(a) (showing only the attributes for $R^\star$).\\

We need to define appropriate distributions in the MG, in such a way that we obtain column $P$ in Figure \ref{fig:result}(b). \ For example, for the first probability in that column, it should be: \ $p_1 = p(\tau_1^1) := P(M_1^m = 1| A = a, \ B=b, M_1^\star = \na)$.

\begin{figure}[h]
\begin{center}\includegraphics[width=3cm]{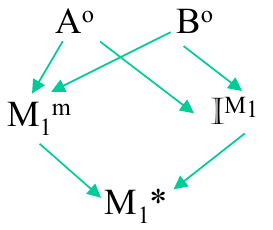}

\vspace{2mm}
{\scriptsize   $\begin{tabu}{c|c|c|c||c|}\hline
R^p&A & B& M_1&P\\ \hline
\tau_1^1 & a & b &1 &p_1\\
\tau_1^2 & a & b&0&(1-p_1)\\ \hhline{~----}
\tau_2^1 & a^\prime & b^\prime &1&p_2\\
\tau_2^2 & a^\prime & b^\prime &0&(1-p_2)\\
\hhline{~----}
\end{tabu}$}\end{center}
\vspace{-4mm}
\caption{(a) Missingness graph. \ \ \ (b) Resulting BID.}\label{fig:result}
\end{figure}

\vspace{1mm}For this MG, the probabilities are calculated as in Example \ref{ex:connection}, obtaining, for example, for the tuples in the first block of $R^p$:
{\scriptsize \begin{eqnarray*}
p(\tau_1^1) &=& P^\mc{M}(M_1^m=1|A=a, B=b) \times P^\mc{M}(A=a,B=b),\\
p(\tau_2^1) &=& P^\mc{M}(M_1^m=0|A=a, B=b) \times P^\mc{M}(A=a,B=b).
\end{eqnarray*}}
Accordingly, we define the probabilities in the MG in such a way that $p(\tau_1^1) = p_1, \ p(\tau_2^1) = (1-p_1)$, etc.

Now, consider the query posed to the original TID: \ $\mc{Q}\!: \ \exists x \exists y(R(x,y) \wedge S(y))$. In order to pose the query to the BID, rewrite it into: \ $\mc{Q}^\prime\!: \ \exists x \exists y(R(x,y,1) \wedge S(y,1))$, and we answer it via the BID.
\boxtheorem
    \end{example}

Now, towards Proposition \ref{lemma:BIDrepresentation}, we first introduce a notion of equivalence between BIDs that ensures their sets of possible worlds, along with the associated probabilities, coincide.

 \begin{definition} (Equivalence of BIDs)
\label{def:bid-equiv}
Let $D$ and $D'$ be two BID instances over the same schema, and let
$\mc{B}$ and $\mc{B}'$ denote their respective sets of blocks.
We say that $D$ and $D'$ are \emph{equivalent}, denoted $D \equiv D'$,
if there exist bijections
\[
f : \mc{B} \to \mc{B}'
\qquad\text{and}\qquad
g : D \to D'
\]
such that, for every block $B \in \mc{B}$:

\begin{enumerate}
    \item[(a)] $g$ restricts to a bijection from $B$ onto $f(B)$;
    \item[(b)] for every tuple $\tau \in B$, we have
    $\tau\!\!\downarrow = g(\tau)\!\!\downarrow$;
    \item[(c)] for every tuple $\tau \in B$,
    \[
    P_B(\tau) = P_{f(B)}\bigl(g(\tau)\bigr),
    \]
    where $P_B(\tau)$ denotes the probability of $\tau$ in block $B$.
\end{enumerate}

When we wish to make the witnessing bijections explicit, we write \ 
$D \stackrel{\tiny f,g}{\equiv} D'$.
\boxtheorem
\end{definition}


As a direct consequence of Definition~\ref{def:bid-equiv}, equivalence of BID instances preserves their possible-world semantics.

\begin{corollary}
\label{cor:equiv-bid}
Let $D$ and $D'$ be two BID instances such that
$D \stackrel{{\tiny f,g}}{\equiv} D'$, and let
$\mc{W}$ and $\mc{W}'$ be their respective sets of possible worlds.
Then there exists a bijection
\[
h : \mc{W} \to \mc{W}'
\]
such that, for every world $W \in \mc{W}$ and every tuple $\tau \in D$:
\begin{enumerate}
    \item[(a)] $\tau \in W$ if and only if $g(\tau) \in h(W)$;
    \item[(b)] $p(W) = p\bigl(h(W)\bigr)$.\boxtheorem
\end{enumerate}
\end{corollary}

Hence, if $D \equiv D'$, then the two instances are indistinguishable with respect to query answering, as every query returns exactly the same answers (with the same probabilities) on both.

Using the notion of equivalence, we show that for every BID instance, there exist an observed instance and an MG whose associated BID is equivalent to the original one.

\begin{proposition}
\label{lemma:BIDrepresentation}
For any BID instance $D$ over a schema $S$, there are an (observed) instance $D^\star$ for a schema $S^\star$ and an MG $\mathcal{M}$, such that, for the BID $D^p$ determined by $D^\star$ and $\mc{M}$, it holds:
$D \equiv D^p$. \boxtheorem
\end{proposition}

\noindent {\bf Proof}: \
Let $D$ be a BID instance over a schema $S$ consisting of $n$ blocks $B_1, \ldots, B_n$, where each bloc $B_i$ contains a set of tuples $t_{i1}, \ldots, t_{1{m_i}}$.
W.l.o.g., we assume that the schema $S$ consists of a single attribute $\mathtt{T}$ that encodes each tuple. Formally, let $S = \{\mathtt{T}\}$, and for every $t_{ij} \in B_i$, we define $t_{ij}[\mathtt{T}] = \mathtt{``t_{ij}"}$.
 Let $p_i^j$ denote the probability of tuple $t_j$ in block $B_i$. By the definition of a BID, we have:
$$
\sum\limits_{j=1}^{m_i} p_i^j = 1 \quad \text{for each } i \in \{1, \ldots, n\}
$$
Figure \ref{fig:arbitrarydBID} shows the BID $D$.
\begin{figure}
    \centering
$\begin{array}{c|c|c|c|}
\hline
\textbf{bloc\_id} & \mathbf{tuple\_id} &  \mathbf{T}  & P   \\ \hline
 B_1 & \tau_1^1  &  t_{11} &   p_1^1\\ \cline{3-4}
  &  \ldots  &  \ldots & \ldots  \\ \cline{3-4}
& \tau_1^{m_1}  &  t_{1{m_1}} &   p_1^{m_1}\\ \hline
  &  \ldots  &  \ldots & \ldots  \\ \hline
 B_n & \tau_n^1  &  t_{n1} &  p_n^1 \\ \cline{3-4}
     &  \ldots  &  \ldots & \ldots  \\ \cline{3-4}
& \tau_n^{m_n}  &  t_{n{m_n}} &   p_n^{m_n}\\ \hline
\end{array}$
\caption{A arbitrary BID $B$}
    \label{fig:arbitrarydBID}
\end{figure}
\begin{figure*}
    \centering
$\begin{array}{c|c|c|c|c|c|c|c|c|c|c|c|c|}
\hline
  \mathbf{tuple\_id} & \mathbf{T^m} & \mathbf{A_1^m} & \ldots & \mathbf{A_n^m}& \mathbf{A_1^*} & \mathbf{\mbb{I}^{A_1}}&\ldots & \mathbf{A_n^*}& \mathbf{\mbb{I}^{A_n}}& \mathbf{T^*} &  \mathbf{\mbb{I}^T} &  \mathbf{Prob} \\ \hline
J^1_{B_1} & t_1 & in(t_1,B_1) & \ldots &   in(t_1,B_n) & 1 & 0 & \ldots & na & 1 & na & 1 & \frac{p^1_1}{n} \\ \cline{2-13}
\ldots & \ldots & \ldots & \ldots & \ldots & \ldots & \ldots & \ldots & \ldots & \ldots & \ldots & \ldots & \ldots\\ \cline{2-13}
J^n_{B_1} & t_n & in(t_n,B_1) & \ldots &   in(t_n,B_n) & 1 & 0 & \ldots & na & 1 & na & 1 & \frac{p^n_1}{n} \\ \cline{2-13}
%
\ldots & \ldots & \ldots & \ldots & \ldots & \ldots & \ldots & \ldots & \ldots & \ldots & \ldots & \ldots & \ldots\\ \cline{2-13}
J^n_{B_n} & t_n & in(t_n, B_n) & \ldots &   in(t_n, B_n) & na & 1 & \ldots & 1 & 0 & na & 1 & \frac{p^n_n}{n} \\ \cline{2-13}
\end{array}$
\caption{The join distribution $P^{\mc{M}}$}
    \label{fig:proof-dist}
\end{figure*}

\begin{figure}
    \centering
$\begin{array}{cc|c|c|c|c|c|}
\hline
 \mathtt{bloc\_id} & \mathtt{tuple\_id} &  \mathbf{T} & \mathbf{A_1} & \ldots & \mathbf{A_n} & P   \\ \hline
 f(\tau_{B_1}) & \tau_{B_1}^1  &  t_1 & 1 & \ldots & A_{n,1} &  p_1^1\\ \cline{3-7}
  &  \ldots  &  \ldots & \ldots & \ldots & \ldots & \ldots  \\ \cline{3-7}
& \tau_{B_1}^{n}  &  t_n & 1 & \ldots & A_{n,n} & p_1^n\\ \hline
\ldots  &  1  &  \ldots & \ldots & \ldots & \ldots & \ldots  \\ \hline
 f(\tau_{B_n}) & \tau_{B_n}^1  &    t_1 & A_{1,1} & \ldots & 1 &  p_n^1\\ \cline{3-7}
  &  \ldots  &  \ldots & \ldots & \ldots & \ldots & \ldots  \\ \cline{3-7}
& \tau_{B_1}^{n}  &  t_n & A_{1,n} & \ldots & 1 & p_n^n\\ \hline
\end{array}$
\caption{The derived BID $D^p$}
    \label{fig:derivedBID}
\end{figure}
To construct the proof, we start with a BID $D$ over a schema $S=\{\mathtt{T}\}$ and, in polynomial time, build an incomplete database $D^*$, an MG $\mc{M}$, and a join distribution $P^{\mc{M}}$. From this construction, we derive a BID instance $D^p \equiv D$.
The database $D^*$ is defined over the schema $\{T^*, A^*_1, \ldots, A^*_n\}$. The attribute $T$ has domain $dom(T) = \{t_1, \ldots, t_n\}$, while each attribute $A_i$ (for $i \in \{1, \ldots, n\}$) has domain $dom(A_i) = [0, 1]$. Each attribute $A_i$ encodes the block $B_i$.
For example, a tuple $t$ in $D^p$ with $t[T] = ``t"$ and $t[A_1] = t[A_3] = 1$ encodes  the fact that the tuple $t$ belongs to blocks $B_1$ and $B_3$.
The database $D^*$ contains $n$ tuples $\tau_{B_i}$, each corresponding to a block $B_i$ in the BID. Each tuple $\tau_{B_i}$ has value $1$  for the attribute $\mathtt{A_i}$ and the value \texttt{na} for all the other  attributes;  that is, $\tau_{B_i}[\mathtt{A_i}] = 1$ and $\tau_{B_i}[\mathtt{T}] = \tau_{B_i}[\mathtt{A_j}] = \texttt{na}$, $\forall j \neq i$.
See Figure~\ref{fig:proof-dstar} for an illustration of $D^*$.
\begin{figure}
    \centering
 $\begin{array}{c|c|c|c|c|}
\hline
\mathbf{bloc\_id} &  \mathbf{T^*} & \mathbf{A^*_1} & \mathbf{\ldots} & \mathbf{A^*_n}   \\ \hline
 \tau_{B_1}  &  na & 1 & na & na  \\ \cline{2-5}
  \ldots  &  \ldots & \ldots & \ldots & \ldots  \\ \cline{2-5}
   \tau_{B_n}  &  na & na & na & 1  \\ \cline{2-5}
\end{array}$
\caption{The incomplete database $D^*$}
    \label{fig:proof-dstar}
\end{figure}
The MG $\mathcal{M}$, shown in Figure~\ref{fig:proof-mg}, encodes the dependency structure between attributes in the incomplete database. Specifically, $\mathcal{M}$ specifies that the value of the attribute $\mathtt{T}$ depends on the values of the attributes $\mathtt{A^*_i}$.







%

\ignore{The MG I have problems with when running Latex outside Overleaf; replaced my figute theMG.pdf
\begin{figure}[h]
    \centering
    \begin{tikzpicture}[->, >=stealth, node distance=0.6cm and 1cm]
        \node (IB1) {$I^{A_1}$};
        \node (IB2) [below=of IB1] {$I^{A_2}$};
        \node (IBdots) [below=of IB2] {$\vdots$};
        \node (IBn) [below=of IBdots] {$I^{A_n}$};

        \node (B1) [right=of IB1] {$A^m_1$};
        \node (B2) [below=of B1] {$A^m_2$};
        \node (Bdots) [below=of B2] {$\vdots$};
        \node (Bn) [below=of Bdots] {$A^m_n$};

        \node (B1star) [below=of IB1,xshift=0.7cm,yshift=0.5cm] {$A^*_1$};
        \node (B2star) [below=of B1star] {$A^*_2$};
        \node (Bstardots) [below=of B2star] {$\vdots$};
        \node (Bnstar) [below=of Bstardots] {$A^*_n$};

        \node (T) [right=of B2] {\texttt{$T^m$}};

        \node (TStar) [right=of T] {\texttt{T}$^*$};
        \node (IT) [above=of TStar] {$I^{\texttt{T}}$};

        \draw[->] (IB1) -- (B1star);
        \draw[->] (IB2) -- (B2star);
        \draw[->] (IBn) -- (Bnstar);

        \draw[->] (B1) -- (B1star);
        \draw[->] (B2) -- (B2star);
        \draw[->] (Bn) -- (Bnstar);

        \draw[->] (B1star) -- (T);
        \draw[->] (B2star) -- (T);
        \draw[->] (Bnstar) -- (T);

        \draw[->] (T) -- (TStar);

        \draw[->] (IT) -- (TStar);

    \end{tikzpicture}
    \caption{The MG $\mathcal{M}$}
    \label{fig:proof-mg}
\end{figure}
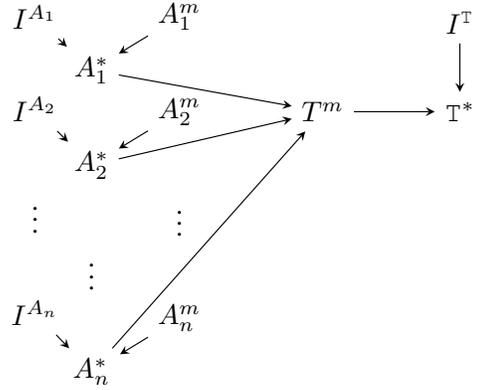
}

\begin{figure}[h]
    \centering
    \includegraphics[width=8cm]{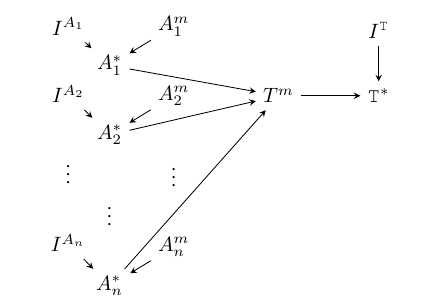}
    \caption{The MG $\mathcal{M}$}
    \label{fig:proof-mg}
\end{figure}

The join distribution $P^{\mathcal{M}}$, illustrated in Figure~\ref{fig:proof-dist}, contains one row $J^j_{B_i}$ for each tuple $t_j$ in block $B_i$. Each row is defined as follows:
\[
\begin{aligned}
& J^j_{B_i}[T^m] = t_j, \\
& J^j_{B_i}[A_l^m] = \text{in}(t_j, B_l),\\ & \quad \text{where } \  \text{in}(t_j, B_l) =
\begin{cases}
0 & \text{if } t_j \in B_l, \\
1 & \text{otherwise},
\end{cases} \\
& J^j_{B_i}[A^*_l] =
\begin{cases}
1 & \text{if } t_j \in B_l, \\
\text{na} & \text{otherwise},
\end{cases}\\
&\quad
J^j_{B_i}[I^{A_l}] =
\begin{cases}
0 & \text{if } t_j \in B_l, \\
1 & \text{otherwise}.
\end{cases}
\end{aligned}
\]

The probability associated with each row is
\[
\Pr(J^j_{B_i}) = \frac{p^j_i}{n}.
\]
These probabilities are normalized, ensuring that
\[
\sum_{i=1}^n \sum_{j=1}^{|B_i|} P^{\mathcal{M}}(J^j_{B_i}) = 1.
\]
The BID instance $D^p$, derived from $(D^*, \mathcal{M},P^{\mathcal{M}})$ is illustrated in Figure~\ref{fig:derivedBID}. It contains a block $f(\tau_{B_i})$ for each incomplete tuple $\tau_{B_i}$ in $D^*$. Each block $f(\tau_{B_i})$ consists of tuples $\tau^j_{B_i}$, which correspond to the possible completions of $\tau_{B_i}$, derived from tuples $t_j$ encoded in $D^p$.

\newpage
The probability $P^{\mathcal{M}}(\tau^j_{B_i})$ assigned to a tuple $\tau^j_{B_i}$ in block $B_i$ of the BID $D^p$ is computed as follows:
\begin{equation}
\begin{aligned}
P^\mathcal{M}(\tau^j_{B_i}) &=
P^\mathcal{M}\Big(
    T^m = t_j \;\Big|\;
    \bigwedge_{l=1}^n A^m_l \in (t_j, B_l),\\
&\quad T^* = \text{na},\\
&\quad A^*_i = 1, \; I^{A_i} = 0,\\
&\quad \bigwedge_{l \neq i} A^*_l = \text{na},\\
&\quad \bigwedge_{l \neq i} I^{A_l} = 1
\Big)
\end{aligned}
\end{equation}
From the MG $\mc{M}$, $T^m$ depends \emph{only} on $A^*_l$, for $l \in [1, n]$.
Hence, we can simplify the conditional probability to:
\begin{equation}
\begin{aligned}
P^\mathcal{M}(\tau^j_{B_i}) &= P^\mathcal{M}(T^m = t_j \mid A^*_i = 1, \bigwedge_{l \neq i} A^*_l = \text{na}) \\
&= \frac{P^\mathcal{M}(T^m = t_j,  A^*_i = 1,  \bigwedge_{l \neq i} A^*_l = \text{na})}{P^\mathcal{M}(A^*_i = 1, \bigwedge_{l \neq i} A^*_l = \text{na})} \\
&=\frac{\frac{p^j_i}{n}}{\sum\limits_{l=1}^n \frac{p^l_i}{n}}=\frac{p^j_i}{\sum\limits_{l=1}^n p^l_i}= p^j_i
\end{aligned}
\end{equation}

The following lemma establishes that the transformation from a BID instance $B$ to the derived instance $D^p$ via the intermediate construction $(D^*, \mathcal{M}, P^{\mathcal{M}})$ preserves the probabilistic semantics over the original schema $S$.
\begin{lemma}
Let $D$ be an arbitrary BID over a schema $S = \{\mathtt{T}\}$. Let $D^*$, $\mathcal{M}$, and the join distribution $P^{\mathcal{M}}$ be constructed from $D$ as described above. Then, we have
$D^p \equiv B$, where  $D^p$ is the BID derived from $D^*$, $\mathcal{M}$, and $P^{\mathcal{M}}$.
\end{lemma}

\vspace{1mm}
\noindent {\bf Proof:} \ Let $B_i$, with $i \in \{1, \ldots, n\}$, denote the blocks in the BID  $D$.
By construction, the mapping $f$ which assigns each block
 $B_i$ in $B$  to the block $f(\tau_{B_i})$ in $D^p$ is bijective. Additionally, there is a one-to-one correspondence between the tuples in each block $B_i$ and those in the corresponding block $\mathcal{B}(\tau_{B_i})$, such that:
\begin{align*}
&\tau_i^j \in B_i \text{ with } \tau_i^j[\mathtt{T}] = \mathtt{"t_j"} \text{ and } p(\tau_i^j) = p_i^j \\
&\qquad\qquad\qquad\qquad \Leftrightarrow \\
&\tau_{B_i}^j \in \mathcal{B}(\tau_{b_i}) \text{ with } \tau_{B_i}^j[\mathtt{T}] = \mathtt{"t_j"} \text{ and } p(\tau_{B_i}^j) = p_i{}^j.
\end{align*}


Hence, we conclude that \ $D \equiv D^p$.\boxtheorem



\subsection*{Proof of Theorem \ref{thm:mcc}.}
The results in Theorem \ref{thm:mcc} are obtained through a  reduction of the \textbf{MCC}-problem to a particular case of ``minimum cost flow problem" in bipartite graphs. The key insight is that finding a most-compliant class corresponds to solving an min cost flow problem where blocks are matched to tuple multiplicities, while minimizing the distance objective.

We formalize the reduction of MCC problem to a minimum cost flow problem.
The reduction constructs a flow network whose feasible integral flows are in one-to-one correspondence with the BID possible worlds, and the cost of a flow equals the compliance distance of the corresponding class-vector.

\paragraph{MCC-Problem Setup.}
Let $D^\star$ be an observed instance, $\mathcal{M}$ its associated MG and $P^\mathcal{M}$ the induced join distribution.
An instance of the MCC problem is given by:
\begin{itemize}
    \item A set of \emph{blocks} $B = \{B_1,\dots,B_n\}$ of the corresponding BID $D^p$.
    \item A set of \emph{tuples} $T = \{t_1,\dots,t_m\}$.
    \item For each block $B_i$ and tuple $t_j$, a probability $p_i(t_j) \in [0,1]$.

\end{itemize}
The goal is to compute an \emph{admissible class-vector} $\mathbf{k}^\star = (k_1^\star,\dots,k_m^\star) \in \mathbb{N}^m$ that minimizes
\[
C_{\mathbf{k}^\star} \in \argmin_{\substack{\mathbf{k}\in\mathbb{N}^m\\ \sum_{j=1}^m k_j = n}} d^c\!\left(C_{\mathbf{k}}, \mathcal{M}\right),
\]

\paragraph{Flow Network Construction.}
From an MCC instance we construct a directed graph $G_{D^p} = (V,E)$ with source $s$ and sink $t$ as follows.

\paragraph{Nodes.}
\[
V = \{s\} \cup L \cup R \cup \{t\},
\]
where
\begin{itemize}
    \item $L = \{u_i \mid B_i \in B\}$ is a set of \emph{block nodes} in one-to-one correspondence with the blocks.
    \item $R = \{v_j \mid t_j \in T\}$ is a set of \emph{tuple nodes} in one-to-one correspondence with the tuples.
\end{itemize}

\paragraph{Arcs.}
The edge set $E = E_S \cup E_{LR} \cup E_T$ consists of three layers:
\begin{enumerate}
    \item \textbf{Source arcs:} For each $u_i \in L$, add an arc $(s,u_i)$ with lower bound $\ell_{s,u_i}=1$ and capacity $c_{s,u_i}=1$.
    \item \textbf{Middle arcs:} For each block $B_i$ and tuple $t_j$ such that $p_i(t_j) > 0$, add an arc $(u_i,v_j)$ with capacity $c_{u_i,v_j}=1$.
    \item \textbf{Sink arcs:} For each $v_j \in R$, add an arc $(v_j,t)$ with capacity $c_{v_j,t} = \deg^{in}(v_j)$, where $\deg^{in}(v_j)$ denotes the number of incoming  arcs to $v_j$ (i.e., the number of blocks that can select $t_j$).
\end{enumerate}

\paragraph{Flow.}
A \emph{feasible integral flow} $f$ assigns to each arc $a\in E$ a nonnegative integer $f_a$ satisfying the lower bounds and capacities, and flow conservation at all intermediate nodes ($L\cup R$).  For the source arcs, the lower bound forces $f_{s,u_i}=1$ for every $i$; thus each block node sends exactly one unit of flow into the network.  By flow conservation, that unit must travel along a middle arc to some tuple node $v_j$ (with $p_i(t_j)>0$) and then to the sink via the corresponding sink arc.  Consequently, the flow values on the sink arcs determine a vector
\[
\mathbf{k}(f) = \bigl(f_{v_1,t},\; f_{v_2,t},\; \dots,\; f_{v_m,t}\bigr) \in \mathbb{N}^m,
\]
which satisfies $\sum_{j=1}^m f_{v_j,t} = n$ because each block contributes exactly one unit to the total outflow from $s$.

\paragraph{Cost.}
The cost of a flow $f$ is defined solely in terms of the vector $\mathbf{k}(f)$:
\[
\operatorname{Cost}(f) = d^c\!\left(C_{\mathbf{k}(f)},\mathcal{M}\right).
\]

We now establish the formal correspondence between feasible integral flows in $G_{D^p}$ and admissible class-vectors, and prove that an optimal flow yields an optimal solution to the MCC problem.

\begin{lemma}
For every feasible integral flow $f$ in $G_{D^p}$, the vector $\mathbf{k}(f)$ is an admissible class-vector (i.e., $\sum_j k_j = n$ and $k_j$ counts the number of blocks assigned to tuple $t_j$).  Conversely, for every admissible class-vector $\mathbf{k}$, there exists a feasible integral flow $f$ such that $\mathbf{k}(f) = \mathbf{k}$.
\end{lemma}

\vspace{1mm}
\begin{proof}
The forward direction follows from the construction: each block $u_i$ sends one unit to some tuple $v_j$, and the sink arc capacities ensure that no tuple receives more than its degree, which is automatically satisfied because each incoming edge to $v_j$ can carry at most one unit.  Hence $\mathbf{k}(f)$ is a vector of nonnegative integers summing to $n$.

Conversely, given $\mathbf{k}$ with $\sum_j k_j = n$, we need to assign each block to a tuple so that exactly $k_j$ blocks are assigned to $t_j$.  This is a feasible assignment if and only if for every tuple $t_j$, the number of blocks that can select $t_j$ (i.e., $\deg(v_j)$) is at least $k_j$.  Since $\mathbf{k}$ arises from a valid assignment in the original MCC instance, this condition holds.  Construct a flow by sending one unit from $s$ to each $u_i$, then from $u_i$ to the tuple node corresponding to its assigned $t_j$, and finally from that tuple node to $t$.  By construction all capacities and lower bounds are satisfied, yielding a feasible integral flow with $\mathbf{k}(f)=\mathbf{k}$.
\end{proof}

\begin{theorem}
An optimal solution to the MCC instance is obtained by solving the minimum cost flow problem on $G_{D^p}$ with objective $\operatorname{Cost}(f)=d^c(C_{\mathbf{k}(f)},\mathcal{M})$.  Specifically, if $f^\star$ minimizes $\operatorname{Cost}(f)$ over all feasible integral flows, then $\mathbf{k}(f^\star)$ is an optimal class-vector for the MCC instance.
\end{theorem}

\vspace{1mm}
\begin{proof}
The lemma establishes a correspondence between feasible integral flows and admissible class-vectors (the mapping is surjective, and two different flows can yield the same $\mathbf{k}$ only if multiple assignments give the same counts, which does not affect the cost because it depends only on $\mathbf{k}$).  The cost of a flow is exactly the compliance distance of the corresponding class-vector.  Therefore minimizing over flows is equivalent to minimizing over class-vectors.
\end{proof}

\paragraph{Complexity Analysis.}
We now analyze the computational complexity of solving the MCC problem via the minimum cost flow reduction. The complexity depends critically on the structure of the distance function $d^c$. When $d^c$ is separable convex, it induces a separable convex cost on the sink arcs of the flow formulation. This places the  problem within the class of minimum cost flow problems with separable convex costs, which    can be solved in polynomial time \cite{Minoux86}.

\begin{lemma}
\label{lem:complexitymcc}
Let  $d^c(\cdot, \mathcal{M})$ be a  convex-separable distance. For an MCC instance with $n$ blocks and $m$ tuples, an optimal admissible class-vector $\mathbf{k}^\star$ can be computed polynomial  time
\end{lemma}

\begin{proof}
By construction, the network $G_{D^p}$ has $N = O(n + m + nm)$ edges and all capacities are integral and polynomially bounded in the input size. Specifically:
\begin{itemize}
    \item Source arcs: $n$ edges with unit capacity.
    \item Middle arcs: at most $n \times m$ edges (only where $p_i(t_j) > 0$) with unit capacity.
    \item Sink arcs: $m$ edges with capacity $\deg(v_j) \le n$.
\end{itemize}
The objective is to minimize a separable convex cost function of the flow on the sink arcs. This is precisely a minimum-cost flow problem with an \emph{edge-separable convex cost function} \cite{Minoux86}.
The algorithm of  \cite{Minoux86} solves such problems  for integral flows on directed graphs in polynomial time.

\end{proof}

\paragraph*{Proof: \ \textbf{\#MCC} is \#P-Complete.}

We prove that \textbf{\#MCC} is \textbf{\#P-complete} via the following two lemmas. Lemma \ref{lemma:mcccountmember} addresses membership in \#P, and Lemma \ref{lemma:mcccounthard} establishes hardness.

\begin{lemma}
\label{lemma:mcccountmember}
Verifying whether a class $C_{\mathbf{k}}$ is an optimal solution to the MCC problem  can be done in polynomial time. Consequently, the problem of counting the number of distinct optimal MCC classes, denoted \#MCC, belongs to the complexity class \#P.
\end{lemma}

\vspace{1mm}
\begin{proof}
We establish polynomial-time verification through two independent checks, each reducible to classical network flow problems.

\paragraph{Step 1: Admissibility.}
Given a candidate vector $\mathbf{k}=(k_1,\dots,k_m)\in\mathbb{N}^m$ with $\sum_{j=1}^m k_j = n$, we must determine whether there exists a feasible assignment of the $n$ blocks to tuples such that exactly $k_j$ blocks are assigned to tuple $t_j$. This is equivalent to checking whether the constructed flow network $G_{D^p}$ admits a feasible integral flow $f$ with $f_{v_j,t}=k_j$ for all $j$.

To test this, construct a modified network $G_{\mathbf{k}}$ from $G_{D^p}$ by setting the capacity of each sink arc $(v_j,t)$ to $k_j$ (all other capacities remain as in $G_{D^p}$: unit capacities on source arcs and middle arcs). A feasible integral flow of value $n$ in $G_{\mathbf{k}}$ exists iff $\mathbf{k}$ is admissible, because such a flow saturates all source arcs (each block sends one unit) and exactly meets the prescribed outflows at the sink.

The network $G_{\mathbf{k}}$ has $O(n+m)$ nodes and $O(nm)$ edges. Computing its maximum flow value can be done in polynomial time using, e.g., Dinic's algorithm \cite{dinic1970}, which runs in $O(\sqrt{n},|E|)$ for unit-capacity bipartite networks. Hence admissibility is decidable in polynomial time.

\paragraph{Step 2: Optimality.}
Let $d_{\min}$ denote the minimum compliance distance for the given MCC instance. By Lemma~\ref{lem:complexitymcc}, an optimal class-vector can be computed in polynomial time (specifically, $O((nm)^{1+o(1)} \log m \log C)$), so $d_{\min}$ can be obtained efficiently. Once $d_{\min}$ is known, verifying whether $C_{\mathbf{k}}$ is optimal reduces to computing its compliance distance $d^c(C_{\mathbf{k}},\mathcal{M})$ and checking whether $d^c(C_{\mathbf{k}},\mathcal{M}) = d_{\min}$. The compliance distance depends only on $\mathbf{k}$ and the fixed instance parameters, and is assumed computable in polynomial time (as is typical for distance functions in such settings). Thus optimality can be checked in polynomial time.

Since both steps run in polynomial time, the decision problem ``Is $\mathbf{k}$ an optimal MCC class?'' belongs to the class \textbf{P} and therefore  \#MCC is in \#P.

\end{proof}

\begin{lemma}
\label{lemma:mcccounthard}
\textbf{\#MCC} is \#P-hard.
\end{lemma}

\vspace{1mm}
\begin{proof}
We prove hardness by a polynomial-time reduction from the problem of counting perfect matchings in 3-regular bipartite graphs, which is known to be \#P-complete. Let \(G = (U \cup V, E)\) be a 3-regular bipartite graph with \(|U| = |V| = n\). We construct a BID \(\mathcal{D}_G^p\) and a join distribution \(P_G^{\mathcal{M}}\) as follows.

\begin{itemize}
    \item \textbf{Tuples.} For each edge \((u_i, v_j) \in E\), create a distinct tuple \(t_{i,j}\). Let \(T = \{ t_{i,j} : (u_i, v_j) \in E \}\) be the set of all tuples. Since \(G\) is 3-regular, \(|U| = |V| = n\) and each vertex has degree 3, we have \(|T| = 3n\).

    \item \textbf{Primary blocks (type \(B\)).} For each vertex \(u_i \in U\), define a block \(b_i\) that contains the tuples corresponding to its incident edges. Specifically, if \(u_i\) is adjacent to \(v_{j_1}, v_{j_2}, v_{j_3}\) in \(G\), then
    \[
    b_i = \{ t_{i,j_1}, t_{i,j_2}, t_{i,j_3} \}.
    \]
    In each such block, the probability distribution is uniform: every tuple \(t_{i,j} \in b_i\) is assigned probability \(p_i(t_{i,j}) = \frac{1}{3}\).

    \item \textbf{Supplementary blocks (type \(S\)).} For each vertex \(v_j \in V\), list its incident tuples in a fixed order, say \(t_{i_1,j}, t_{i_2,j}, t_{i_3,j}\) where \((u_{i_1}, v_j), (u_{i_2}, v_j), (u_{i_3}, v_j) \in E\). Add two supplementary blocks \(s_{j,1}\) and \(s_{j,2}\). The block \(s_{j,1}\) contains the two consecutive tuples \(t_{i_1,j}\) and \(t_{i_2,j}\), and the block \(s_{j,2}\) contains the two consecutive tuples \(t_{i_2,j}\) and \(t_{i_3,j}\). In each supplementary block, the probability distribution is uniform: every tuple in the block has probability \(\frac{1}{2}\). Consequently, each tuple \(t_{i,j}\) belongs to exactly one primary block \(b_i\) and to either one or two supplementary blocks: specifically, \(t_{i_1,j}\) belongs only to \(s_{j,1}\); \(t_{i_2,j}\) belongs to both \(s_{j,1}\) and \(s_{j,2}\); and \(t_{i_3,j}\) belongs only to \(s_{j,2}\).

    \item \textbf{Join distribution.} The join distribution \(P_G^{\mathcal{M}}\) is uniform over all tuples, i.e., \(P_G^{\mathcal{M}}(t) = \frac{1}{3n}\) for every \(t \in T\).
\end{itemize}

This construction is clearly computable in polynomial time in the size of \(G\).

We consider a convex separable distance function of the form
\[
d(C_{\mathbf{k}}, \mathcal{M}) = \sum_{j=1}^{m} d\bigl(k_j, P^{\mathcal{M}}(t_j)\bigr),
\]
where \(k_j\) denotes the number of occurrences of tuple \(t_j\) in the multiset \(C_{\mathbf{k}}\) (i.e., in the class). We define \(d\bigl(k_j, P^{\mathcal{M}}(t_j)\bigr)\) to be any convex function satisfying
\[
d\bigl(k_j, P^{\mathcal{M}}(t_j)\bigr) = 0 \quad\text{if } k_j \in \{0,1\},\]
\[d\bigl(k_j, P^{\mathcal{M}}(t_j)\bigr) > 0 \quad\text{if } k_j \ge 2.
\]
For instance, we can define
\[
d(k,p) =
\begin{cases}
0, & \text{if } k \le 1,\\[4pt]
k-1, & \text{if } k \ge 2,
\end{cases}
\]
which is convex in \(k\) (its discrete second differences are nonnegative) and independent of \(p\) (as \(p\) is uniform in our construction).

In our setting, the join distribution \(P_G^{\mathcal{M}}\) is uniform over all tuples, so the contribution of each tuple depends only on its occurrence count. Consequently, a class \(C_{\mathbf{k}}\) with \(\mathbf{k} = (k_1,\dots,k_{3n})\) such that every \(k_j \le 1\)—call such a class a \emph{0/1 class}—has zero cost and is therefore optimal whenever it is realizable.

The  structure of the supplementary blocks imposes a crucial constraint on 0/1 classes: for each vertex \(v_j \in V\), among its three incident tuples \(t_{i_1,j}, t_{i_2,j}, t_{i_3,j}\), at most one can be selected by a primary block. Indeed, if two of these tuples were selected by primary blocks, then at least one supplementary block would be forced to select a tuple already used, creating an occurrence count of \(k_j \ge 2\) and thus positive cost. Moreover, if exactly one of these tuples is selected by a primary block, the supplementary blocks \(s_{j,1}\) and \(s_{j,2}\) have a unique valid assignment using the remaining two tuples, each exactly once. Conversely, if no tuple among the three is selected by a primary block, the supplementary blocks cannot both select distinct tuples without violating the 0/1 condition. Hence, in any realizable 0/1 class, for each \(v_j\) exactly one of its three incident tuples is selected by a primary block, and the other two are assigned to the supplementary blocks in the unique forced manner.

This property establishes a bijection between the set of 0/1 classes and the set of perfect matchings of the original 3‑regular bipartite graph \(G\). Specifically, each primary block \(b_i\) (corresponding to \(u_i \in U\)) selects exactly one tuple, corresponding to an edge incident to \(u_i\). The condition that each \(v_j\) has exactly one of its incident tuples selected by a primary block ensures that the selected edges form a matching that covers all vertices in \(V\); since \(|U| = |V| = n\) and each \(u_i\) selects exactly one edge, this matching is perfect. Thus every perfect matching gives rise to a unique 0/1 class, and conversely every 0/1 class yields a perfect matching.

A crucial observation is that, by the structure of the construction and the definition of the distance function, either \emph{all} minimum‑cost classes are 0/1 classes, or \emph{none} of them are. This follows from the fact that the cost of any class is determined solely by the occurrence counts: if a 0/1 class exists, it attains cost zero, which is the global minimum, and any class with a tuple occurring more than once incurs a positive cost and therefore cannot be an MCC. Conversely, if no 0/1 class exists, then every realizable class has at least one tuple with \(k_j \ge 2\), yielding positive cost, and all such classes have the same minimum cost (by uniformity of the construction). Hence the number of minimum‑cost classes (MCCs) of \((\mathcal{D}_G^p, P_G^{\mathcal{M}})\) equals either the number of perfect matchings in \(G\) (when perfect matchings exist) or some constant independent of the graph's matching structure (when no perfect matching exists). Since counting perfect matchings in 3‑regular bipartite graphs is \#P‑complete, we conclude that \#MCC is \#P‑hard.
\end{proof}

\paragraph*{Proof: \textbf{MCC-Enum} is in \textbf{DelayP}.}

We outline a naive polynomial-delay algorithm for enumerating all the MCCs.
The algorithm explores the vector space of MCCs in lexicographic order using a recursive depth-first search.

The search space can be represented as a rooted tree:
\begin{itemize}
    \item The root is labeled by the empty vector $\emptyset$.
    \item The first level consists of nodes labeled $1,2,\ldots,m$.
    \item In general, the children of a node labeled by $h$ are nodes labeled $h0, h1, \ldots, hn$.  Note that, at a level $l$, only \emph{correct} vectors are generated: namely those satisfying the constraint:$\sum_{j=1}^l h_j \;\leq\; n$.
\end{itemize}


A node at depth $l$ encodes the vector $v = (k_1,\ldots,k_l)$ corresponding to the set of classes whose characteristic vectors begin with the prefix $(k_1,\ldots,k_l)$.
In particular, the leaves of the tree correspond bijectively to complete MCC vectors, i.e., to the classes themselves.










The enumeration algorithm generates every MCC exactly once.
Thus, the crux of ensuring polynomial delay lies in testing, at each step, whether an intermediate node is \emph{valid}, i.e., whether it can lead to at least one MCC, before exploring its children.
We now detail this validity check.

\medskip
First, compute the global minimum cost $c_{\min}$ of an optimal flow $f$ $G_{D^p}$
by solving a min cost flow in $G_{D^p}$ (in polynomial time).

An intermediate node $h = (k_1,\ldots,k_l)$ corresponds to the partial assignment where tuple $t_j$ ($1 \leq j \leq l$) is used exactly $k_j$ times.
We then construct a new graph $G_{D^p}^{/h}$ from $G_{D^p}$ by assigning to each arc $(v_j, t)$, for $1 \le j \le l$, a lower bound and capacity given by
$\ell_{v_j,t} = c_{v_j,t} = k_j$.
This enforces that each $v_j$, for $1 \le j \le l$, must be matched to exactly $k_j$ block nodes in $L$, thereby preserving only matchings consistent with $h$.
A node $h$ is valid if and only if $G_{D^p}^{/h}$  admits a minimal cost  flow $G_{D^p}^{/h}$  equal to $c_{\min}$.

\medskip
Both the construction of $G_{D^p}^{/h}$ and the computation of the optimal cost can be performed in polynomial time.
Overall, the enumeration visits at most $O(mn)$ nodes between two consecutive MCCs.
\end{document}